\documentclass[11pt,a4paper]{article}
\usepackage[T1]{fontenc}
\usepackage[utf8]{inputenc}
\usepackage{lmodern}
\usepackage{microtype}
\usepackage[
  a4paper,
  top=27mm,
  bottom=27mm,
  left=27mm,
  right=27mm
]{geometry}
\usepackage{amsmath,amssymb,amsthm,mathtools}
\usepackage{mathrsfs}
\usepackage{bm}
\usepackage{booktabs}
\usepackage{tabularx}
\usepackage{enumitem}
\usepackage{xcolor}
\usepackage{tikz}
\usepackage{float}
\usepackage{aliascnt}
\usepackage{authblk}
\usepackage[
  colorlinks=true,
  linkcolor=blue!55!black,
  citecolor=blue!55!black,
  urlcolor=blue!55!black
]{hyperref}

\hypersetup{
  pdftitle={Almost One Bit Violation of Minimum-Output Renyi Entropy Additivity Simultaneously at All Orders},
  pdfauthor={Guocheng Zhen, Chengkai Zhu, Ranyiliu Chen, and Xin Wang},
  pdfsubject={Commuting free factors, exact Bell quotients, and a common near-one-bit Renyi gap},
  pdfkeywords={minimum output Renyi entropy, additivity violation, free groups,
    Haagerup inequality, Bell quotient, strong convergence, random orthogonal matrices}
}

\usepackage[nameinlink,noabbrev]{cleveref}
\setlist{nosep}
\allowdisplaybreaks[1]
\theoremstyle{plain}
\newtheorem{theorem}{Theorem}[section]
\newaliascnt{proposition}{theorem}
\newtheorem{proposition}[proposition]{Proposition}
\aliascntresetthe{proposition}
\newaliascnt{lemma}{theorem}
\newtheorem{lemma}[lemma]{Lemma}
\aliascntresetthe{lemma}
\newaliascnt{corollary}{theorem}
\newtheorem{corollary}[corollary]{Corollary}
\aliascntresetthe{corollary}

\theoremstyle{definition}
\newaliascnt{definition}{theorem}
\newtheorem{definition}[definition]{Definition}
\aliascntresetthe{definition}
\newaliascnt{question}{theorem}

\aliascntresetthe{question}

\theoremstyle{remark}
\newaliascnt{remark}{theorem}

\aliascntresetthe{remark}

\crefname{theorem}{Theorem}{Theorems}
\Crefname{theorem}{Theorem}{Theorems}
\crefname{proposition}{Proposition}{Propositions}
\Crefname{proposition}{Proposition}{Propositions}
\crefname{lemma}{Lemma}{Lemmas}
\Crefname{lemma}{Lemma}{Lemmas}
\crefname{corollary}{Corollary}{Corollaries}
\Crefname{corollary}{Corollary}{Corollaries}
\crefname{definition}{Definition}{Definitions}
\Crefname{definition}{Definition}{Definitions}
\crefname{question}{Question}{Questions}
\Crefname{question}{Question}{Questions}
\crefname{remark}{Remark}{Remarks}
\Crefname{remark}{Remark}{Remarks}

\newcommand{\C}{\mathbb C}
\newcommand{\R}{\mathbb R}
\newcommand{\N}{\mathbb N}
\newcommand{\Z}{\mathbb Z}
\newcommand{\F}{\mathbb F}

\newcommand{\id}{\mathrm{id}}
\newcommand{\Tr}{\operatorname{Tr}}

\newcommand{\supp}{\operatorname{supp}}
\newcommand{\rank}{\operatorname{rank}}
\newcommand{\diag}{\operatorname{diag}}
\newcommand{\vecop}{\operatorname{vec}}

\newcommand{\Dens}{\mathcal D}
\newcommand{\B}{\mathcal B}
\newcommand{\U}{\mathcal U}
\newcommand{\cE}{\mathcal E}
\newcommand{\cQ}{\mathcal Q}
\newcommand{\cS}{\mathcal S}

\newcommand{\sG}{\mathsf G}

\newcommand{\Hminp}{H_{\min,p}}
\newcommand{\HS}{\mathrm{HS}}

\newcommand{\Ran}{\operatorname{Ran}}
\newcommand{\ket}[1]{\lvert #1\rangle}
\newcommand{\bra}[1]{\langle #1\rvert}
\newcommand{\proj}[1]{\lvert #1\rangle\!\langle #1\rvert}

\newcolumntype{Y}{>{\raggedright\arraybackslash}X}

\title{\textbf{Almost One Bit Violation of Minimum-Output R\'enyi
Entropy Additivity Simultaneously at All Orders}}

\author[1]{Guocheng Zhen\thanks{zhenguocheng0814@gmail.com}}
\author[2]{Chengkai Zhu\thanks{zhuchengkai7@gmail.com}}
\author[3]{Ranyiliu Chen\thanks{chenranyiliu@quantumsc.cn}}
\author[1]{Xin Wang\thanks{felixxinwang@hkust-gz.edu.cn}}
\affil[1]{\small Thrust of Artificial Intelligence, Information Hub,\par The Hong Kong University of Science and Technology (Guangzhou), Guangdong 511453, China}
\affil[2]{\small QudeLeap Research, Shanghai 200030, China}
\affil[3]{\small Quantum Science Center of Guangdong-Hong Kong-Macao Greater Bay Area, Shenzhen 518045, China}
\date{}

\begin{document}

\maketitle

\begin{abstract}
We prove that minimum-output R\'enyi-entropy additivity can fail by almost one bit simultaneously at every nonnegative order. For every $\varepsilon\in(0,\log2)$, there exists a finite-dimensional quantum channel with a real Stinespring isometry such that the same maximally entangled input witnesses a tensor-square entropy gap of at least $\log2-\varepsilon$ for all $p\in[0,\infty]$. The output dimension can be chosen to be $O(\varepsilon^{-3})$ as $\varepsilon\downarrow0$.

The construction uses direct products of free groups: tensorized Haagerup estimates control the one-copy outputs, while commutation between distinct factors forces exact Bell-branch collisions at two copies. Strong convergence gives both an existential realization through finite-dimensional representations of right-angled Artin groups followed by realification, and a Haar-orthogonal model whose success probability tends to one as the matrix dimension grows. We also determine the exact Bell quotient, prove asymptotically sharp regular-radius bounds, and show that the cubic output-dimension scale is optimal within the present purity--rank certificate.
\end{abstract}

\tableofcontents

\section{Introduction}\label{sec:introduction}

For finite-dimensional quantum channels $\Phi$ and $\Psi$, product inputs
always give
\[
  \Hminp(\Phi\otimes\Psi)
  \le
  \Hminp(\Phi)+\Hminp(\Psi).
\]
The minimum-output-entropy additivity problem asks whether equality must hold for every pair of channels.  At the von Neumann point
$p=1$, this question is particularly important because of its connection with classical communication over quantum channels.  The Holevo--Schumacher--Westmoreland theorem gives the regularized Holevo formula for the unassisted classical capacity, while Shor proved that,
as universal statements over finite-dimensional quantum channels, additivity of minimum output von Neumann entropy is equivalent to additivity of the Holevo quantity and to the corresponding additivity
and strong-superadditivity statements for entanglement of formation
\cite{SchumacherWestmoreland1997,Holevo1998,Shor2004,ShorErratum2004}.

Counterexamples were first obtained away from the von Neumann point.
Werner and Holevo gave an explicit violation for sufficiently large
R\'enyi orders \cite{WernerHolevo2002}, and Hayden and Winter
subsequently proved nonadditivity for every $p>1$ by random
constructions \cite{HaydenWinter2008}.  Hastings then settled the
von Neumann case $p=1$ by a random finite-dimensional construction
\cite{Hastings2009}.  Subsequent work clarified the probabilistic,
geometric, and random-matrix mechanisms behind these examples
\cite{FukudaKingMoser2010,BrandaoHorodecki2010,
AubrunSzarekWerner2011,BelinschiCollinsNechita2012,
BelinschiCollinsNechita2016,CollinsFukudaZhong2015}.  In particular,
Belinschi, Collins, and Nechita showed that the Bell-state violation for
conjugate random channels can be made arbitrarily close to one bit at
$p=1$ \cite[Theorem~6.3]{BelinschiCollinsNechita2016}.  Constructive
counterexamples were later obtained for every $p>2$
\cite{GrudkaHorodeckiPankowski2010}, and more recently for every
$p>1$ \cite{DerksenLovitz2026}.

Collins gave an especially economical operator-algebraic proof of
von Neumann-entropy nonadditivity for mixed-unitary channels: strong
asymptotic freeness replaces large Haar matrices by free Haar
unitaries, while Haagerup's inequality controls the resulting
length-two free-group polynomials
\cite{CollinsMale2014,Collins2018}.  Particularly relevant to the
present work, Collins and Youn subsequently developed a
Haagerup inequality for products of free groups and used commuting
families associated with such products to prove nonadditivity of a
regularized minimum output entropy in an infinite-dimensional
commuting-operator setting \cite{CollinsYoun2022}; see also
\cite{KalantarShobeiri2025} for further developments in that setting.

For orders below one, Cubitt, Harrow, Leung, Montanaro, and Winter proved nonadditivity at $p=0$, and hence for sufficiently small positive
orders \cite{CubittEtAl2008}.  More recently, Leung, Lovitz, and Wu established finite-dimensional counterexamples for $0 \le p < \frac{1}{4}$ and $p > \frac{3}{4}$ using Haar-random projection-induced channels \cite{LeungLovitzWu2026}.  Their analysis employs two correlated two-copy witnesses: a transpose-complement rank defect at low orders and a product--conjugate Bell output at higher orders.

In this work, we introduce a channel model based on the direct product $\mathbb F_d^{\,r}$ of free groups and prove an almost-one-bit violation simultaneously at all R\'enyi orders. More precisely, for every $\varepsilon\in(0,\log 2)$ there exist positive integers $m,K$ and a channel $\Phi:M_m(\mathbb C)\to M_K(\mathbb C)$ with a real Stinespring isometry such that
\[
  2H_{\min,p}(\Phi)
  -
  H_p\!\left(
    \Phi^{\otimes 2}
    \bigl(|\Omega_m\rangle\langle\Omega_m|\bigr)
  \right)
  \ge \log 2-\varepsilon,
  \qquad p\in[0,\infty].
\]
The witnessing input is the maximally entangled state on two copies of the input space, and neither the channel nor the witness depends on the R\'enyi order $p$. The mechanism is built into the direct-product group structure. Tensorized Haagerup estimates within the free factors bound the quadratic branch radius, which in turn controls both the one-copy purity and the largest output eigenvalue. Exact commutation between distinct factors creates a large family of exact Bell collisions, producing a Bell quotient whose rank ratio can be made arbitrarily close to one half. For $0\le p\le2$, the one-copy purity estimate together with this rank defect gives the required entropy comparison. At higher orders, we retain the largest-eigenvalue control together with the full Bell fibre-weight distribution and interpolate their moments. These complementary estimates yield, for the same channel and the same Bell input, a gap arbitrarily close to one bit uniformly over all $p\in[0,\infty]$.

To realize this mechanism in finite dimensions, we use strong convergence in two complementary ways. Finite-dimensional representations of right-angled Artin groups, followed by realification, give a non-random existence result, while independent Haar orthogonal families acting on separate tensor factors give a random realization whose success probability tends to one. Moreover, we obtain several sharp quantitative estimates that clarify the asymptotic strength of the construction and the sharpness of the underlying estimates.

\section{Main results}\label{sec:main-results}

We first introduce the direct-product channel model and state the
simultaneous almost-one-bit theorem, followed by a finite-output example
and a non-random existence corollary. Throughout,
logarithms are natural, so one bit is $\log2$. 

\subsection{The direct-product channel model}

Fix integers $r,d\ge2$ and let
\begin{equation}\label{eq:group-branches}
  \begin{aligned}
    \sG_{r,d}&:=\prod_{j=1}^r\F_d^{(j)},
      &K&:=2rd,\\
    \cS_j&:=\{g_{j,i}^{\pm1}:1\le i\le d\}\quad(1\le j\le r),
      &\cS&:=\bigsqcup_{j=1}^r\cS_j.
  \end{aligned}
\end{equation}
The $j$th family freely generates $\F_d^{(j)}$, and generators in different families commute. The branch set is inverse-closed and the number $r$ counts
factors within one channel.

For a real orthogonal representation $\pi:\sG_{r,d}\to O(m)$, write
$W_s:=\pi(s)$ and define the isometry
\begin{equation}\label{eq:stinespring}
  V_\pi:\C^m\longrightarrow\C^m\otimes\C^{\cS},\qquad
  V_\pi x:=\frac1{\sqrt K}\sum_{s\in\cS}W_sx\otimes\ket s.
\end{equation}
Its complementary channels are
\begin{align}
  \cE_\pi(X)&:=\frac1K\sum_{s\in\cS}W_sXW_s^*,
       \label{eq:mixed-channel}\\
  \Phi_\pi(X)&:=\frac1K
       \bigl[\Tr(W_sXW_t^*)\bigr]_{s,t\in\cS}.
       \label{eq:complementary-channel}
\end{align}
The first is a uniform mixed-orthogonal channel, and the second has
output dimension $K$. Since the matrices $W_s$ are real, entrywise
channel conjugation fixes $\Phi_\pi$:
\[
  \overline\Phi(X):=\overline{\Phi(\overline X)},\qquad
  \overline{\Phi_\pi}=\Phi_\pi.
\]
Thus the usual product--conjugate Bell input is a genuine self-tensor witness.

For each $N$, let $(O_{j,i}^{(N)})_{1\le j\le r,\,1\le i\le d}$ be
mutually independent Haar-distributed matrices in $O(N)$. Set
\begin{equation}\label{eq:random-representation}
  \pi_N(g_{j,i}):=
  I_N^{\otimes(j-1)}\otimes O_{j,i}^{(N)}
                   \otimes I_N^{\otimes(r-j)}.
\end{equation}
This is an exact representation of $\sG_{r,d}$ at every finite $N$.
We write
\[
  \Phi_N:=\Phi_{\pi_N}:M_{N^r}(\C)\longrightarrow M_K(\C).
\]
No relation between the random choices at different values of $N$ is required.

For a channel $\Phi:M_m(\C)\to M_K(\C)$, use the Bell vector
\begin{equation}\label{eq:bell-vector}
  \ket{\Omega_m}:=\frac1{\sqrt m}\sum_{a=1}^m\ket a\otimes\ket a
\end{equation}
and define
\begin{equation}\label{eq:gap-definition}
  \begin{aligned}
    Z_\Phi&:=\Phi^{\otimes2}(\proj{\Omega_m}),\\
    \Delta_p^\Omega(\Phi)&:=2H_{\min,p}(\Phi)-H_p(Z_\Phi).
  \end{aligned}
\end{equation}
This Bell-witness gap is a lower bound for the minimum-output entropy
defect:
\begin{equation}\label{eq:witness-to-additivity-gap}
  2H_{\min,p}(\Phi)-H_{\min,p}(\Phi^{\otimes2})
  \ge\Delta_p^\Omega(\Phi).
\end{equation}
For the branch set in \eqref{eq:group-branches}, put
\begin{equation}\label{eq:MrK-definition}
  M_{r,K}:=\frac{r+1}{2r}K^2-K+1.
\end{equation}
We will show that this is the number of distinct relative operations
$s^{-1}t$ and, with probability tending to one, the exact Bell-output
rank.

\subsection{Simultaneous almost-one-bit nonadditivity}

\begin{theorem}[A common almost-one-bit gap]\label{thm:main}
For every $0<\varepsilon<\log2$, there exist integers $r,d\ge2$ such that,
with $K=2rd$, the Haar-orthogonal model satisfies
\begin{equation}\label{eq:main-probability}
  \Pr\!\left[
    \log2-\varepsilon
       \le\inf_{p\in[0,\infty]}\Delta_p^\Omega(\Phi_N)
       \le\log\frac{K^2}{M_{r,K}}
  \right]\longrightarrow1
  \qquad(N\to\infty).
\end{equation}
Here $\log(K^2/M_{r,K})<\log2$, and the parameters can be chosen so that
\[
  K=O(\varepsilon^{-3})\qquad(\varepsilon\downarrow0).
\]
In particular, one finite-dimensional channel with a real Stinespring
isometry and one Bell input give a gap at least $\log2-\varepsilon$
simultaneously at every R\'enyi order.
\end{theorem}

The parameters $r,d$ and hence $K$ are chosen first, depending on $\varepsilon$ but not on $p$. They remain fixed while $N$ tends to infinity. A sufficiently large finite realization is then chosen once for all orders. The lower bound in \eqref{eq:main-probability} is a common nonadditivity witness; the upper bound concerns this same Bell witness, not the potentially larger minimum-output defect in \eqref{eq:witness-to-additivity-gap}.

\subsection{A finite-output example and a non-random consequence}

The following parameter choice gives a concrete uniform all-order violation. It is included only as an illustrative example, and no optimality is claimed for either the output dimension or the resulting gap.

\begin{corollary}[An all-order gap at output dimension $320$]
\label{cor:numerical}
For $r=2$, $d=80$, and $K=320$, let
\begin{equation}\label{eq:numerical-gap-constant}
  \gamma_{320}:=
  \log\frac{320^4}{76481\,(370+1/100)^2}
  >1.4\times10^{-3}.
\end{equation}
Then the Haar-orthogonal model satisfies
\begin{equation}\label{eq:numerical-gap}
  \Pr\!\left[
    \inf_{p\in[0,\infty]}\Delta_p^\Omega(\Phi_N)\ge\gamma_{320}
  \right]\longrightarrow1.
\end{equation}
\end{corollary}

Finally, $\sG_{r,d}$ is the right-angled Artin group of the complete
$r$-partite graph with $d$ vertices in each part. The strong approximation
theorem of Magee and Thomas therefore gives an additional existence
consequence of the same entropy argument.

\begin{corollary}[Non-random existence via right-angled Artin groups]
\label{cor:nonrandom}
For every $0<\varepsilon<\log2$, the parameters $r,d,K$ in
\Cref{thm:main} also admit a sequence of real orthogonal representations
\[
  \pi_n:\sG_{r,d}\longrightarrow O(m_n)
\]
strongly converging to the regular representation such that, for every
sufficiently large $n$,
\begin{equation}\label{eq:nonrandom-gap}
  \log2-\varepsilon
  \le\inf_{p\in[0,\infty]}\Delta_p^\Omega(\Phi_{\pi_n})
  \le\log\frac{K^2}{M_{r,K}}<\log2.
\end{equation}
The output-dimension choice remains $K=O(\varepsilon^{-3})$.
At $r=2,d=80$, a strongly convergent real representation sequence also
satisfies the bound $\gamma_{320}$ for every sufficiently large term.
\end{corollary}

Here ``non-random'' describes the existence of a representation sequence. No closed-form matrices or algorithm for producing that sequence are supplied.

\section{Preliminaries}\label{sec:preliminaries}

This section mainly records the operator-algebraic, quantum-information, and group-theoretic facts used later.  General references for operator algebras
and free probability are \cite{BrownOzawa2008,VoiculescuDykemaNica1992,NicaSpeicher2006, MingoSpeicher2017}.

\subsection{Operator and quantum-information conventions}

All Hilbert spaces are finite-dimensional unless explicitly stated
otherwise.  Throughout, Hilbert-space inner products are conjugate-linear
in the first argument and linear in the second.  For a Hilbert space
$\mathcal H$, let $\B(\mathcal H)$ denote the bounded operators on
$\mathcal H$.

The symbol $\Tr$ denotes the unnormalized matrix trace.  For a matrix
$X$, we use the operator norm
\[
  \|X\|
  :=
  \sup_{\|v\|_2=1}\|Xv\|_2
\]
and the Hilbert--Schmidt norm
\[
  \|X\|_{\HS}
  :=
  \bigl(\Tr(X^*X)\bigr)^{1/2}.
\]
In particular, if $X\ge0$, then
\[
  \|X\|=\lambda_{\max}(X).
\]

The state space of $\mathcal H$ is
\[
  \Dens(\mathcal H)
  :=
  \{\rho\in\B(\mathcal H):\rho\ge0,\ \Tr\rho=1\}.
\]
A quantum channel
\[
  \Phi:\B(\mathcal H_{\mathrm{in}})
  \longrightarrow
  \B(\mathcal H_{\mathrm{out}})
\]
is a completely positive trace-preserving linear map.

For $0<p<\infty$, $p\ne1$, the R\'enyi entropy is
\[
  H_p(\rho)
  :=
  \frac{1}{1-p}\log\Tr(\rho^p),
\]
with the continuous and endpoint conventions
\[
  \begin{aligned}
    H_0(\rho)&:=\log\rank\rho,\\
    H_1(\rho)&:=-\Tr(\rho\log\rho),\\
    H_\infty(\rho)&:=-\log\|\rho\|
                     =-\log\lambda_{\max}(\rho),
  \end{aligned}
\]
where $0\log0:=0$.  For a fixed state, these are the corresponding
limits as $p\downarrow0$, $p\to1$, and $p\to\infty$. The minimum output R\'enyi entropy of $\Phi$ is
\[
  \Hminp(\Phi)
  :=
  \min_{\rho\in\Dens(\mathcal H_{\mathrm{in}})}
  H_p\bigl(\Phi(\rho)\bigr).
\]
The minimum is attained for every $p\in[0,\infty]$: for
$p\in(0,\infty]$ the entropy is continuous on the finite-dimensional
state space, while at $p=0$ the rank, and hence $H_0$, is lower
semicontinuous.

For a finite probability vector $q=(q_1,\ldots,q_\ell)$, we use the
shorthand
\[
  H_p(q)
  :=
  H_p\!\left(\operatorname{diag}(q_1,\ldots,q_\ell)\right),
\]
so that $H_p(q)$ denotes the R\'enyi entropy of the corresponding diagonal density matrix. The normalized Bell vector is defined in \eqref{eq:bell-vector}.
With the column-vectorization convention
$\vecop(X)=\sum_jXe_j\otimes e_j$, one has
\begin{equation}\label{eq:vec-identity}
  (A\otimes B)\ket{\Omega_m}
  =
  \frac1{\sqrt m}\vecop(AB^T).
\end{equation}

\begin{lemma}[Elementary R\'enyi bounds]\label{lem:renyi-bounds}
Let $\rho$ be a density matrix of rank $q$.
\begin{enumerate}[label=(\roman*)]
  \item The function $p\mapsto H_p(\rho)$ is nonincreasing on
  $[0,\infty]$.
  \item For every $p\in[0,\infty]$,
  \[
    H_p(\rho)\le\log q.
  \]
  \item For every $0\le p\le2$,
  \[
    H_p(\rho)\ge H_2(\rho)=-\log\Tr(\rho^2).
  \]
\end{enumerate}
\end{lemma}

\begin{proof}
Let $\lambda=(\lambda_1,\ldots,\lambda_q)$ be the positive eigenvalues
of $\rho$.  Part~(i) follows from monotonicity of the weighted power
means
\[
  \left(\sum_{i=1}^q\lambda_i\lambda_i^s\right)^{1/s},
  \qquad s=p-1,
\]
with the orders $p=0,1,\infty$ obtained by the corresponding limits.
The uniform distribution on $q$ points is majorized by every probability
vector on $q$ points, and R\'enyi entropy is Schur-concave; this gives
(ii).  Part~(iii) follows from (i).
\end{proof}

\begin{lemma}[Complementary reductions of a pure state]
\label{lem:complementary-spectrum}
Let $V:\mathcal H\to\mathcal K\otimes\mathcal L$ be an isometry, and let
$\cE$ and $\Phi$ be the channels obtained by tracing over $\mathcal L$
and $\mathcal K$, respectively.  For every unit vector $x\in\mathcal H$,
the states $\cE(\proj{x})$ and $\Phi(\proj{x})$ have the same nonzero
eigenvalues, including multiplicities.  The same statement applies to
$V^{\otimes2}$ and every pure two-copy input.
\end{lemma}

\begin{proof}
For a unit vector $x\in\mathcal H$, the vector
$Vx\in\mathcal K\otimes\mathcal L$ is pure, and
\[
  \cE(\proj{x})
  =
  \Tr_{\mathcal L}\proj{Vx},
  \qquad
  \Phi(\proj{x})
  =
  \Tr_{\mathcal K}\proj{Vx}.
\]
The two reduced states of a bipartite pure state have the same nonzero
eigenvalues, including multiplicities, by the Schmidt decomposition.

For a pure two-copy input, apply the same argument to
$(V\otimes V)y$ after the canonical reordering
\[
  \mathcal K\otimes\mathcal L\otimes\mathcal K\otimes\mathcal L
  \cong
  (\mathcal K\otimes\mathcal K)\otimes
  (\mathcal L\otimes\mathcal L).
\]
The corresponding reduced states are
$(\cE\otimes\cE)(\proj{y})$ and
$(\Phi\otimes\Phi)(\proj{y})$, respectively.
\end{proof}

\subsection{Free groups and the tensorized Haagerup inequality}

A free group $\F_d=\langle g_1,\ldots,g_d\rangle$ consists of reduced
words in the letters $g_i^{\pm1}$. Every nonidentity element has a unique
expression
\[
  g_{i_1}^{\varepsilon_1}\cdots g_{i_\ell}^{\varepsilon_\ell},
  \qquad \varepsilon_a\in\{+1,-1\},
\]
with no adjacent cancellation. Its reduced word length is $\ell$.
For the direct product $\sG_{r,d}$ in \eqref{eq:group-branches}, every
element has a unique coordinate representation
$g=(g_1,\ldots,g_r)$, with $g_j\in\F_d^{(j)}$. If $|g_j|$ denotes
reduced length in the $j$th factor, define the multidegree
\[
  \deg(g):=(|g_1|,\ldots,|g_r|),\qquad
  S_{\boldsymbol\ell}:=
  \{g\in\sG_{r,d}:\deg(g)=\boldsymbol\ell\}
  \quad(\boldsymbol\ell\in\Z_{\ge0}^r).
\]
Unlike a free product, this direct product permits cancellation only within a coordinate; letters in different coordinates commute.

For a discrete group $G$, its left regular representation is
\[
  \lambda_G:G\longrightarrow\U(\ell^2(G)),\qquad
  \lambda_G(g)\delta_h:=\delta_{gh},
\]
where $(\delta_h)_{h\in G}$ is the canonical orthonormal basis. The
reduced group $C^*$-algebra is
\[
  C_{\mathrm r}^*(G):=C^*(\lambda_G(G))\subseteq\B(\ell^2(G)).
\]
Its canonical trace is $\tau(T)=\langle\delta_e,T\delta_e\rangle$ and
satisfies $\tau(\lambda_G(g))=\mathbf1_{\{g=e\}}$. For a finitely
supported function $f:G\to\C$, write
\[
  \lambda_G(f):=\sum_{g\in G}f(g)\lambda_G(g),\qquad
  \|f\|_2:=\left(\sum_g|f(g)|^2\right)^{1/2}.
\]
Then $\tau(\lambda_G(f)^*\lambda_G(f))^{1/2}=\|f\|_2$.
For $G=\sG_{r,d}$ we abbreviate $\lambda_G$ by $\lambda_{r,d}$.
Under the natural identification
\[
  \ell^2(\sG_{r,d})\cong\bigotimes_{j=1}^r\ell^2(\F_d^{(j)}),
\]
this representation acts coordinatewise.

Haagerup's free-group inequality states that
$\|\lambda_{\F_d}(f)\|\le(\ell+1)\|f\|_2$ when $f$ is supported on
reduced words of length $\ell$ \cite[Lemma~1.4]{Haagerup1979}.
For products of free groups, Collins and Youn established the
corresponding multidegree estimate
\cite[Proposition~3.2]{CollinsYoun2022}. The scalar form used here also
follows from \cite[Lemma~9.4]{BordenaveCollins2024}.

\begin{lemma}[Tensorized Haagerup inequality]\label{lem:tensor-haagerup}
If $f:\sG_{r,d}\to\C$ is finitely supported and
$\supp f\subseteq S_{\boldsymbol\ell}$, then
\begin{equation}\label{eq:tensor-haagerup}
  \|\lambda_{r,d}(f)\|
  \le\prod_{j=1}^r(\ell_j+1)\|f\|_2.
\end{equation}
For $1\le j\le r$, let $e_j\in\mathbb N_0^r$ denote the $j$th standard basis vector.  In particular, the Haagerup constants for the multidegrees $2e_j$ and $e_j+e_a$ $(j\ne a)$ are $3$ and $4$, respectively.
\end{lemma}

\subsection{Right-angled Artin groups and complete multipartite graphs}

For a finite simple graph $\Lambda=(V,E)$, its right-angled Artin group is
\[
  A_\Lambda:=\left\langle
    v\in V\ \middle|\ [v,w]=e\text{ whenever }\{v,w\}\in E
  \right\rangle.
\]
Thus an edge imposes commutation, while a non-edge imposes no relation.
Let $\Lambda_{r,d}$ be the complete $r$-partite graph with vertex classes
$V_j=\{g_{j,1},\ldots,g_{j,d}\}$, $1\le j\le r$. It has every edge
between distinct classes and no edge within a class. Consequently,
\begin{equation}\label{eq:raag-direct-product}
  A_{\Lambda_{r,d}}
  =\left\langle g_{j,i}\ \middle|\
    [g_{j,i},g_{a,b}]=e\text{ for }j\ne a\right\rangle
  \cong\prod_{j=1}^r\F_d^{(j)}=\sG_{r,d}.
\end{equation}
The example $r=d=3$ is shown in \Cref{fig:complete-multipartite}.
The defining graph has $rd$ generator vertices, whereas the
inverse-closed channel branch set has $K=2rd$ elements.

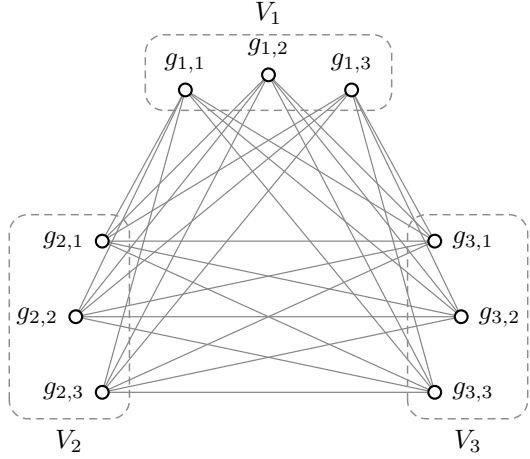
\begin{figure}[htbp]
  \centering
  \begin{tikzpicture}[
    x=1cm,y=1cm,
    edge/.style={draw=black!48,line width=0.45pt},
    part/.style={draw=black!45,densely dashed,rounded corners=7pt,
                 line width=0.55pt},
    vertex/.style={circle,draw=black,fill=white,line width=0.8pt,
                   inner sep=0pt,minimum size=4.8pt},
    every node/.style={font=\small}
  ]
    \coordinate (A1) at (-1.10, 2.25);
    \coordinate (A2) at ( 0.00, 2.45);
    \coordinate (A3) at ( 1.10, 2.25);
    \coordinate (B1) at (-2.20, 0.25);
    \coordinate (B2) at (-2.55,-0.75);
    \coordinate (B3) at (-2.20,-1.75);
    \coordinate (C1) at ( 2.20, 0.25);
    \coordinate (C2) at ( 2.55,-0.75);
    \coordinate (C3) at ( 2.20,-1.75);
    \foreach \i in {1,2,3}{
      \foreach \j in {1,2,3}{
        \draw[edge] (A\i)--(B\j);
        \draw[edge] (A\i)--(C\j);
        \draw[edge] (B\i)--(C\j);
      }
    }
    \draw[part] (-1.63,1.98) rectangle (1.63,2.98);
    \draw[part] (-3.43,-2.10) rectangle (-1.84,0.60);
    \draw[part] ( 1.84,-2.10) rectangle ( 3.43,0.60);
    \foreach \i in {1,2,3}{
      \node[vertex,label=above:{$g_{1,\i}$}] at (A\i) {};
      \node[vertex,label=left:{$g_{2,\i}$}] at (B\i) {};
      \node[vertex,label=right:{$g_{3,\i}$}] at (C\i) {};
    }
    \node at (0,3.27) {$V_1$};
    \node at (-2.64,-2.39) {$V_2$};
    \node at ( 2.64,-2.39) {$V_3$};
  \end{tikzpicture}
  \caption{The complete tripartite graph
  $\Lambda_{3,3}=K_{3,3,3}$, illustrating
  $\Lambda_{r,d}=K_{d,\ldots,d}$. Each dashed part contains three
  generators and no internal edge; every pair of vertices in distinct
  parts is joined by an edge imposing commutation. Thus
  $A_{\Lambda_{3,3}}\cong\F_3^3$. Inverse generators are not drawn
  as separate vertices.}
  \label{fig:complete-multipartite}
\end{figure}

This graph description records exactly the two features used in the proof: freeness inside each class and commutation across classes. It also allows the strong-approximation theorem for right-angled
Artin groups to be applied directly to $\sG_{r,d}$, rather than constructing approximations factor by factor and separately verifying strong convergence of the resulting tensor-product representations.

\subsection{Strong convergence and finite-dimensional approximation}

\begin{definition}[Strong convergence to the regular representation]
\label{def:strong-convergence}
Let $G$ be a discrete group. A sequence of finite-dimensional unitary
representations $\rho_n:G\to\U(m_n)$ strongly converges to the left
regular representation if, for every $z=\sum_gc_gg\in\C[G]$,
\[
  \|\rho_n(z)\|\longrightarrow\|\lambda_G(z)\|,
  \qquad \rho_n(z):=\sum_gc_g\rho_n(g).
\]
For random representations, strong convergence in probability means
that this convergence holds in probability for every fixed $z\in\C[G]$.
\end{definition}

The main realization uses the tensor-product strong-convergence theorem of Bordenave and Collins.  We use the Haar-orthogonal specialization noted there, applied to the representation in \eqref{eq:random-representation} for fixed $r$ and $d$.

\begin{lemma}[Haar-orthogonal tensor approximation
{\cite[Theorem~9.2]{BordenaveCollins2024}}]\label{lem:bordenave-collins}
For each fixed $r,d\ge2$, the representations $\pi_N$ in
\eqref{eq:random-representation} strongly converge in probability to
$\lambda_{r,d}$ as $N\to\infty$.
\end{lemma}

This is a strong-convergence theorem for the full tensor model, not a consequence of tensoring the norm limits of individual factors. For the supplementary existence result, \eqref{eq:raag-direct-product} allows us to use the following theorem.

\begin{lemma}[Finite-dimensional approximation of $\sG_{r,d}$
{\cite[Theorem~1.1]{MageeThomas2026}}]\label{lem:magee-thomas}
For every fixed $r,d\ge2$, there exists a sequence of finite-dimensional
unitary representations
\[
  \rho_n:\sG_{r,d}\longrightarrow\U(\ell_n)
\]
that strongly converges to $\lambda_{r,d}$.
\end{lemma}

Realification converts this sequence to real orthogonal representations
without changing its limiting norms, as follows.

\begin{lemma}[Realification preserves strong convergence]
\label{lem:realification}
Let $G$ be a discrete group.  If
$\rho_n:G\to\U(m_n)$ strongly converges to the left regular
representation $\lambda_G$, then the realifications
\[
  \rho_n^{\R}(g)
  :=
  \begin{pmatrix}
    \Re\rho_n(g)&-\Im\rho_n(g)\\
    \Im\rho_n(g)& \Re\rho_n(g)
  \end{pmatrix}
  \in O(2m_n)
\]
also strongly converge to $\lambda_G$.
\end{lemma}

\begin{proof}
Realification is multiplicative and sends unitary matrices to real
orthogonal matrices, so $\rho_n^{\R}$ is an orthogonal representation
of $G$.  After complexification, a unitary change of basis independent
of $g$ gives
\[
  \rho_n^{\R}(g)
  \simeq
  \rho_n(g)\oplus\overline{\rho_n(g)}.
\]
Consequently, for
\[
  z=\sum_g c_g g\in\C[G],
  \qquad
  \overline z:=\sum_g\overline{c_g}\,g,
\]
one has
\[
  \rho_n^{\R}(z)
  \simeq
  \rho_n(z)\oplus
  \overline{\rho_n(\overline z)}.
\]
Therefore
\[
  \|\rho_n^{\R}(z)\|
  =
  \max\bigl\{
    \|\rho_n(z)\|,
    \|\rho_n(\overline z)\|
  \bigr\}.
\]
Both terms converge by the assumed strong convergence.  Moreover, in
the canonical basis of $\ell^2(G)$ every $\lambda_G(g)$ is a real
permutation operator, so
\[
  \lambda_G(\overline z)
  =
  \overline{\lambda_G(z)},
  \qquad
  \|\lambda_G(\overline z)\|
  =
  \|\lambda_G(z)\|.
\]
Hence
\[
  \|\rho_n^{\R}(z)\|
  \longrightarrow
  \|\lambda_G(z)\|,
\]
which proves the claimed strong convergence.
\end{proof}

\section{A finite-dimensional entropy criterion}
\label{sec:finite-criterion}

This section isolates the finite-dimensional entropy mechanism behind the construction, independently of how the required channel is later realized.  We first show that suitable operator-norm control keeps all one-copy outputs sufficiently close to the maximally mixed state.  We then analyze the maximally entangled two-copy input, for which the algebraic relations among the branches create exact collisions and reduce the output entropy.  Comparing the two sides gives a criterion for simultaneous Rényi nonadditivity at all orders.

\subsection{Mixed-unitary channels and the quadratic radius}

For arbitrary unitaries $W_1,\ldots,W_K\in\U(m)$, define $\Phi$
by the formula in \eqref{eq:complementary-channel}, with indices
$1,\ldots,K$, and put
\begin{equation}\label{eq:finite-quadratic-radius}
  \begin{aligned}
    \cQ_W(A)&:=\sum_{i,j=1}^K a_{ij}W_i^*W_j,\\
    \Gamma_K(W)&:=
      \sup_{\substack{A=A^*\in M_K(\C),\ \Tr A=0\\\|A\|_{\HS}=1}}
             \|\cQ_W(A)\|.
  \end{aligned}
\end{equation}
Direct calculation gives
\begin{equation}\label{eq:duality-identity}
  \Tr(A\Phi(X))=\frac1K\Tr(X\cQ_W(A)).
\end{equation}
The branch Gram matrix in the convention adapted to this identity is
\begin{equation}\label{eq:pure-output-gram}
  [G_x]_{ij}:=\langle W_jx,W_ix\rangle,
  \qquad \Phi(\proj{x})=G_x/K
  \quad(\|x\|=1).
\end{equation}
It is positive semidefinite with diagonal one.
For traceless $A$ the diagonal contribution to $\cQ_W(A)$ vanishes.
Thus the supremum defining $\Gamma_K(W)$ is unchanged if $a_{ii}=0$
is imposed. By \Cref{lem:renyi-bounds}, an upper bound on purity already gives a
one-copy entropy lower bound throughout $0\le p\le2$.

\begin{lemma}[Exact Gram radius and maximum purity
{\cite[Theorem~4.1]{ZhenZhuChenWang2026}}]
\label{lem:exact-radius}
For every finite unitary tuple $W=(W_1,\ldots,W_K)$ and the channel $\Phi$ in \eqref{eq:complementary-channel},
\begin{equation}\label{eq:exact-radius}
  \Gamma_K(W)
  =\max_{\|x\|=1}\|G_x-I_K\|_{\HS}
  =K\max_{X\in\Dens(\C^m)}
       \left\|\Phi(X)-\frac{I_K}{K}\right\|_{\HS}.
\end{equation}
Moreover,
\begin{equation}\label{eq:maximum-purity}
  \max_{X\in\Dens(\C^m)}\Tr\bigl(\Phi(X)^2\bigr)
  =\frac1K+\frac{\Gamma_K(W)^2}{K^2}.
\end{equation}
\end{lemma}

\begin{proof}
For a Hermitian traceless $A$, the operator $\cQ_W(A)$ is self-adjoint.
By \eqref{eq:duality-identity} and \eqref{eq:pure-output-gram},
\[
  \|\cQ_W(A)\|
  =\max_{\|x\|=1}|\langle x,\cQ_W(A)x\rangle|
  =\max_{\|x\|=1}|\Tr(A(G_x-I_K))|.
\]
Since $G_x-I_K$ is Hermitian and traceless, Hilbert--Schmidt duality on
the real Hilbert space of traceless Hermitian matrices gives
\[
  \sup_{\substack{A=A^*,\ \Tr A=0\\\|A\|_{\HS}=1}}
  |\Tr(A(G_x-I_K))|=\|G_x-I_K\|_{\HS}.
\]
Interchanging the two suprema proves the first equality in
\eqref{eq:exact-radius}.  The second follows for pure inputs from
\eqref{eq:pure-output-gram}, and for mixed inputs by convexity of the
Hilbert--Schmidt norm.  Finally, for every output state $\sigma$,
\[
  \Tr(\sigma^2)
  =\frac1K+\left\|\sigma-\frac{I_K}{K}\right\|_{\HS}^2.
\]
Maximizing and using \eqref{eq:exact-radius} proves
\eqref{eq:maximum-purity}.
\end{proof}

The preceding identities immediately convert a bound on the quadratic radius into the one-copy entropy estimate needed below.

\begin{corollary}[One-copy bounds at all orders]
\label{cor:one-copy}
Suppose $\Gamma_K(W)\le C$ and put
\begin{equation}\label{eq:ab-constants}
  a:=1+C^2/K,\qquad b:=1+C.
\end{equation}
Every output $\sigma=\Phi(X)$ satisfies
\begin{equation}\label{eq:purity-peak}
  \Tr\sigma^2\le a/K,\qquad \|\sigma\|\le b/K.
\end{equation}
Consequently,
\begin{align}
  H_{\min,p}(\Phi)&\ge\log K-\log a,
          &&0\le p\le2,\label{eq:one-copy-low}\\
  H_{\min,p}(\Phi)&\ge\log K-
    \frac{\log a+(p-2)\log b}{p-1},
          &&2\le p<\infty,\label{eq:one-copy-high}\\
  H_{\min,\infty}(\Phi)&\ge\log K-\log b.
          &&\label{eq:one-copy-infinity}
\end{align}
\end{corollary}

\begin{proof}
The purity estimate is \eqref{eq:maximum-purity}, and
\[
  \|\sigma\|\le\frac1K+\|\sigma-I_K/K\|
  \le\frac1K+\|\sigma-I_K/K\|_{\HS}\le b/K.
\]
Order monotonicity proves \eqref{eq:one-copy-low}. For $p\ge2$,
the eigenvalues give the moment interpolation
\[
  \Tr\sigma^p\le\|\sigma\|^{p-2}\Tr\sigma^2
        \le K^{1-p}ab^{p-2}.
\]
Take logarithms and multiply by $1/(1-p)<0$ to obtain
\eqref{eq:one-copy-high}. The norm bound gives the endpoint
\eqref{eq:one-copy-infinity}. Each estimate holds for every input,
so it remains valid after minimization.
\end{proof}

\subsection{The exact Bell quotient}\label{sec:exact-bell-quotient}

The structural statement behind \Cref{thm:main} applies to a more general
branch set. Let $G$ be a discrete group, let $S=S^{-1}\subset G$ consist
of $K$ distinct elements, and let $\pi:G\to O(m)$ be a real orthogonal
representation. Define $V_\pi$, $\Phi_\pi$, and $Z_{\Phi_\pi}$ by the same
formulas as above, now indexed by $S$. Introduce the relative-operation
map
\begin{equation}\label{eq:relative-operation-map}
  \begin{aligned}
    \mu &: S\times S\longrightarrow G,
       &\mu(s,t)&:=s^{-1}t,\\
    \mathcal T&:=S^{-1}S,
       &n_g&:=|\mu^{-1}(g)|\quad(g\in\mathcal T).
  \end{aligned}
\end{equation}
The fibres determine a probability vector and an isometry:
\begin{equation}\label{eq:quotient-isometry}
  \begin{aligned}
    q_g&:=\frac{n_g}{K^2},
      &u_g&:=\frac1{\sqrt{n_g}}
                 \sum_{\mu(s,t)=g}\ket{s,t},\\
    U&:\C^{\mathcal T}\longrightarrow\C^S\otimes\C^S,
      &U\ket g&:=u_g.
  \end{aligned}
\end{equation}
Since the fibres of $\mu$ are pairwise disjoint, for $g,h\in\mathcal T$,
\[
\begin{aligned}
\langle u_g,u_h\rangle
&=
\frac{1}{\sqrt{n_gn_h}}
\sum_{\substack{\mu(s,t)=g\\ \mu(s',t')=h}}
\langle s,t\,|\,s',t'\rangle  \\
&=
\delta_{g,h}.
\end{aligned}
\]
Thus $(u_g)_{g\in\mathcal T}$ is an orthonormal family, and hence $U^*U=I_{\C^{\mathcal T}}$; in particular, $U$ is an isometry. The following proposition separates the group-theoretic quotient
structure of the Bell output from the representation-dependent overlaps, encoded respectively by $(U,D_q)$ and $R_\pi$.

\begin{proposition}[Bell quotient and exact rank]\label{prop:bell-rank-defect}
Let $G,S,\pi$ be as above. Set
\[
  D_q:=\diag(q),\qquad
  (R_\pi)_{g,h}:=\frac1m\Tr\pi(h^{-1}g),\qquad
  B_\pi:=D_q^{1/2}R_\pi D_q^{1/2},
\]
and let $P_S\ket s:=\ket{s^{-1}}$. Then $R_\pi$ is a real correlation matrix, $\diag(B_\pi)=q$, and
\begin{equation}\label{eq:exact-bell-quotient}
  \widetilde Z_\pi:=(P_S\otimes P_S)Z_{\Phi_\pi}(P_S\otimes P_S)^*
       =U B_\pi U^*.
\end{equation}
Consequently,
\begin{equation}\label{eq:bell-quotient-rank-kernel}
  \begin{aligned}
    \rank Z_{\Phi_\pi}&=\rank R_\pi\le|\mathcal T|,\\
    \ker U^*&\subseteq\ker\widetilde Z_\pi,
      \qquad \dim\ker U^*=K^2-|\mathcal T|.
  \end{aligned}
\end{equation}
After zero padding, $q$ is majorized by the eigenvalue vector of
$Z_{\Phi_\pi}$. In particular,
\begin{equation}\label{eq:bell-quotient-majorization}
  H_p(Z_{\Phi_\pi})\le H_p(q)\quad(p\in[0,\infty]),
  \qquad \|Z_{\Phi_\pi}\|\ge\frac1K.
\end{equation}
\end{proposition}

\begin{proof}
For $g\in\mathcal T$, set
$\eta_g:=m^{-1/2}\vecop(\pi(g))$. These are real unit vectors, and
\begin{equation}\label{eq:Rpi-gram}
  \langle\eta_h,\eta_g\rangle
  =\frac1m\Tr\!\bigl(\pi(h)^*\pi(g)\bigr)
  =\frac1m\Tr\pi(h^{-1}g)
  =(R_\pi)_{g,h}.
\end{equation}
Thus $R_\pi$ is a real correlation matrix, so $B_\pi\ge0$ and
\begin{equation}\label{eq:Bpi-diagonal}
  \diag(B_\pi)=q,\qquad \Tr B_\pi=1.
\end{equation}

After the canonical reordering that places the two system registers
before the two branch registers, the Stinespring output is
\begin{equation}\label{eq:bell-stinespring-state}
  \Psi_\pi
  =\frac1K\sum_{s,t\in S}
    (W_s\otimes W_t)\ket{\Omega_m}\otimes\ket{s,t}.
\end{equation}
Define
$\widetilde\Psi_\pi:=(I_{m^2}\otimes P_S\otimes P_S)\Psi_\pi$.
Since $S=S^{-1}$ and $W_{s^{-1}}=W_s^*$, relabelling both branches
gives
\begin{equation}\label{eq:bell-relabelled-state}
  \widetilde\Psi_\pi
  =\frac1K\sum_{s,t\in S}
    (W_s^*\otimes W_t^*)\ket{\Omega_m}\otimes\ket{s,t}.
\end{equation}
By \eqref{eq:vec-identity} and real orthogonality,
\begin{equation}\label{eq:bell-relative-vector}
  (W_s^*\otimes W_t^*)\ket{\Omega_m}
  =\frac1{\sqrt m}\vecop(W_s^*W_t)
  =\eta_{\mu(s,t)}.
\end{equation}
Grouping the terms by the fibres of $\mu$ therefore yields
\begin{equation}\label{eq:bell-fibre-decomposition}
  \widetilde\Psi_\pi
  =\frac1K\sum_{s,t\in S}\eta_{\mu(s,t)}\otimes\ket{s,t}
  =\sum_{g\in\mathcal T}\sqrt{q_g}\,\eta_g\otimes u_g.
\end{equation}
Taking the system partial trace and using \eqref{eq:Rpi-gram}, we obtain
\begin{equation}\label{eq:bell-reduction-expanded}
  \begin{aligned}
    \widetilde Z_\pi
    &=\Tr_{\C^m\otimes\C^m}\proj{\widetilde\Psi_\pi}\\
    &=\sum_{g,h\in\mathcal T}
      \sqrt{q_gq_h}\,(R_\pi)_{g,h}\ket{u_g}\!\bra{u_h}
      =UB_\pi U^*.
  \end{aligned}
\end{equation}
This proves \eqref{eq:exact-bell-quotient}.

Since $U$ is an isometry, $UB_\pi U^*$ is unitarily equivalent to
$B_\pi\oplus0$ on $\Ran U\oplus\ker U^*$. Moreover, every $q_g$ is
positive, so $D_q$ is invertible and
\begin{equation}\label{eq:Bpi-rank}
  \rank Z_{\Phi_\pi}
  =\rank B_\pi
  =\rank R_\pi
  \le|\mathcal T|.
\end{equation}
The factorization also gives
$\ker U^*\subseteq\ker\widetilde Z_\pi$, while
$\dim\ker U^*=K^2-|\mathcal T|$. This proves
\eqref{eq:bell-quotient-rank-kernel}.

By \eqref{eq:Bpi-diagonal} and the Schur--Horn theorem
\cite{Bhatia1997},
\begin{equation}\label{eq:q-majorized-by-spectrum}
  q\prec\lambda(B_\pi).
\end{equation}
The eigenvalues of $Z_{\Phi_\pi}$ are those of $B_\pi$ together with
$K^2-|\mathcal T|$ zeros, which proves the claimed majorization after
zero padding. Schur concavity gives
$H_p(Z_{\Phi_\pi})\le H_p(q)$ for $0<p<\infty$. The endpoint $p=0$
follows from the rank bound and the positivity of every $q_g$; the
endpoint $p=\infty$ follows from
$\|Z_{\Phi_\pi}\|=\|B_\pi\|\ge\max_g q_g$.
Finally, $\mu(s,t)=e$ exactly when $s=t$, so $n_e=K$ and
\[
  \|Z_{\Phi_\pi}\|
  =\|B_\pi\|
  \ge q_e
  =\frac1K.
\]
This proves \eqref{eq:bell-quotient-majorization}.
\end{proof}

\subsection{The fibre profile for commuting free factors}

We now classify completely the fibres of the relative-operation map
\[
  \mu:\cS\times\cS\longrightarrow\sG_{r,d},
  \qquad
  \mu(s,t):=s^{-1}t.
\]
Equivalently, we determine how the $K^2$ ordered branch pairs
$(s,t)\in\cS\times\cS$ are grouped according to their common relative
operation: how many distinct group labels occur, and how many
preimages each label has.  Since every branch lies in a single free
factor, a nonidentity relative operation has nontrivial entries in
either one or two coordinates.

\begin{lemma}[Relative-operation fibres]\label{lem:fibres}
Let
\[
  k:=|\cS_j|=2d=\frac{K}{r}.
\]
The fibres of
\[
  \mu:\cS\times\cS\to\sG_{r,d},
  \qquad
  \mu(s,t)=s^{-1}t,
\]
are exactly of the following three types:
\begin{itemize}
  \item \emph{Identity fibre.}
  The identity element has the single fibre
  \[
    \mu^{-1}(e)=\{(s,s):s\in\cS\},
  \]
  of size $K$.

  \item \emph{Same-factor fibres.}
  If $s,t\in\cS_j$ with $s\ne t$, then the fibre of
  $s^{-1}t$ is the singleton $\{(s,t)\}$.  There are
  \[
    r\,k(k-1)
    =
    \frac{K^2}{r}-K
  \]
  such fibres.

  \item \emph{Mixed fibres.}
  If $s\in\cS_j$ and $t\in\cS_\ell$ with $j\ne\ell$, then
  \begin{equation}\label{eq:mixed-fibre}
    \mu^{-1}(s^{-1}t)
    =
    \{(s,t),(t^{-1},s^{-1})\}.
  \end{equation}
  Hence every mixed fibre has size two, and the number of mixed
  fibres is
  \[
    \frac{r(r-1)k^2}{2}
    =
    \frac{r-1}{2r}K^2.
  \]
\end{itemize}
These three classes are disjoint and exhaust all fibres.  Consequently,
\begin{equation}\label{eq:bell-rank-count}
  |\cS^{-1}\cS|
  =
  1+\left(\frac{K^2}{r}-K\right)
  +\frac{r-1}{2r}K^2
  =
  \frac{r+1}{2r}K^2-K+1
  =
  M_{r,K}.
\end{equation}
\end{lemma}

\begin{proof}
We treat the three cases separately.

\begin{itemize}
  \item \emph{Identity fibre.}
  One has
  \[
    \mu(s,t)=e
    \quad\Longleftrightarrow\quad
    s^{-1}t=e
    \quad\Longleftrightarrow\quad
    s=t.
  \]
  Hence the identity fibre consists exactly of the $K$ diagonal
  pairs $(s,s)$.

  \item \emph{Same-factor fibres.}
  Suppose $s,t\in\cS_j$ and $s\ne t$.  Then $s^{-1}t$ is a reduced
  word of length two in the free group $\mathbb F_d^{(j)}$.
  By uniqueness of reduced words, the two letters determine $s$ and
  $t$ uniquely.  The nontrivial coordinate also determines the factor
  $j$.  Thus the fibre is the singleton $\{(s,t)\}$. For each factor there are $k(k-1)$ ordered off-diagonal pairs, so
  the total number of same-factor singleton fibres is
  \[
    r\,k(k-1)
    =
    \frac{K^2}{r}-K.
  \]

  \item \emph{Mixed fibres.}
  Suppose $s\in\cS_j$ and $t\in\cS_\ell$ with $j\ne\ell$.
  Since different free factors commute,
  \[
    s^{-1}t
    =
    ts^{-1}
    =
    (t^{-1})^{-1}s^{-1},
  \]
  and therefore both $(s,t)$ and $(t^{-1},s^{-1})$ belong to the same fibre. Conversely, the two nontrivial coordinates of $s^{-1}t$ determine the two factors and the corresponding letters uniquely.  Hence these are the only two preimages, proving
  \eqref{eq:mixed-fibre}. There are
  \[
    r(r-1)k^2
  \]
  ordered cross-factor pairs.  Since each mixed fibre contains two
  such pairs, the number of mixed fibres is
  \[
    \frac{r(r-1)k^2}{2}
    =
    \frac{r-1}{2r}K^2.
  \]
\end{itemize}

Every ordered pair $(s,t)\in\cS\times\cS$ belongs to exactly one of
the three cases above, so the classification is exhaustive.  Summing
the numbers of fibres gives \eqref{eq:bell-rank-count}.
\end{proof}

The preceding fibre count can now be translated into the spectral and
entropy information needed for the Bell output.

\begin{corollary}[Exact quotient and Bell entropy bounds]
\label{cor:bell-profile}
Define
\begin{equation}\label{eq:bell-weights}
  q_{r,K}:=\left(
    \frac1K,
    \left(\frac2{K^2}\right)^{\!\times (r-1)K^2/(2r)},
    \left(\frac1{K^2}\right)^{\!\times(K^2/r-K)}
  \right),
\end{equation}
where $v^{\times n}$ denotes $n$ repetitions.  For every real
orthogonal representation of $\sG_{r,d}$,
\begin{align}
  \rank Z_{\Phi_\pi}
  &=
  \rank R_\pi
  \le
  M_{r,K},
  \notag\\
  \dim\ker U^*
  &=
  \frac{r-1}{2r}K^2+K-1,
  \label{eq:bell-rank}\\
  H_p(Z_{\Phi_\pi})
  &\le
  H_p(q_{r,K})
  \le
  \log M_{r,K},
  \qquad p\in[0,\infty].
  \label{eq:bell-entropy}
\end{align}
The first entropy bound retains the full fibre profile and is therefore
sharper than the rank-only bound $\log M_{r,K}$.  Its endpoint and
order-one values are
\begin{equation}\label{eq:bell-endpoints}
  \begin{aligned}
    H_0(q_{r,K})
    &=
    \log M_{r,K},\\
    H_1(q_{r,K})
    &=
    2\log K-\frac{\log K}{K}
    -\frac{r-1}{r}\log 2,\\
    H_\infty(q_{r,K})
    &=
    \log K.
  \end{aligned}
\end{equation}
\end{corollary}

\begin{proof}
Apply \Cref{prop:bell-rank-defect} to the fibre classification in
\Cref{lem:fibres}.  Dividing the three fibre sizes
\(
  K, 2, 1
\)
by the total number $K^2$ of branch pairs gives the probability vector
\eqref{eq:bell-weights}.  The rank and kernel statements then follow
from \eqref{eq:bell-quotient-rank-kernel}, while
\eqref{eq:bell-entropy} follows from
\eqref{eq:bell-quotient-majorization}.

The endpoint and order-one formulas are obtained directly from
\eqref{eq:bell-weights}.  In particular, the total weight carried by
the mixed fibres is
\[
  \frac{r-1}{2r}K^2\cdot\frac{2}{K^2}
  =
  \frac{r-1}{r},
\]
which gives the mixed-fibre contribution
$-\frac{r-1}{r}\log 2$ in $H_1(q_{r,K})$.
\end{proof}

\subsection{One certificate for the entire R\'enyi scale}

The one-copy purity bound and the Bell rank estimate give the
low-order certificate
\begin{equation}\label{eq:low-certificate}
  D_{r,K,C}:=
  \log\frac{K^2}{M_{r,K}(1+C^2/K)^2}.
\end{equation}
It applies throughout $0\le p\le2$. Above order two, the same one-copy lower bound is no longer available, so we retain both the full Bell moment and the largest-eigenvalue estimate.

\begin{theorem}[All-order finite-dimensional criterion]
\label{thm:finite-criterion}
Let $\pi:\sG_{r,d}\to O(m)$ be a real orthogonal representation.
Suppose its branch tuple
\[
  W=(\pi(s))_{s\in\cS}
\]
satisfies $\Gamma_K(W)\le C$.  With $a,b$ as in
\eqref{eq:ab-constants}, set
\begin{equation}\label{eq:beta-alpha}
  \beta_r:=\frac{2(r-1)}{r},
  \qquad
  \alpha_r:=\log4-\frac34\log3+\frac14\log\beta_r,
\end{equation}
and define
\begin{equation}\label{eq:all-order-certificate}
  \delta_{r,K,C}
  :=
  \min\left\{
    D_{r,K,C},\
    \alpha_r-2\log a,\
    \frac34\log K+\frac14\log2-2\log b
  \right\}.
\end{equation}
Then
\begin{equation}\label{eq:finite-common-gap}
  \inf_{p\in[0,\infty]}
  \Delta_p^\Omega(\Phi_\pi)
  \ge
  \delta_{r,K,C}.
\end{equation}
In particular, if $\delta_{r,K,C}>0$, then the same channel
$\Phi_\pi$ and the same Bell input certify strict self-tensor
nonadditivity simultaneously for every $p\in[0,\infty]$.
\end{theorem}

\begin{proof}
We treat the low-order range, the higher-order range, and the
endpoint $p=\infty$ separately, and then combine the resulting
estimates into a single uniform bound.

For $0\le p\le2$, \eqref{eq:one-copy-low} gives
\[
  H_{\min,p}(\Phi_\pi)
  \ge
  \log K-\log a.
\]
On the other hand, \eqref{eq:bell-entropy} gives
\[
  H_p(Z_{\Phi_\pi})
  \le
  \log M_{r,K}.
\]
Therefore
\begin{align}
  \Delta_p^\Omega(\Phi_\pi)
  &=
  2H_{\min,p}(\Phi_\pi)-H_p(Z_{\Phi_\pi})
  \notag\\
  &\ge
  2\log K-2\log a-\log M_{r,K}
  \notag\\
  &=
  \log\frac{K^2}{M_{r,K}(1+C^2/K)^2}
  =
  D_{r,K,C},
  \qquad 0\le p\le2.
  \label{eq:finite-low-gap}
\end{align}

We next consider $p=2+u$ with $u\ge0$.  By
\eqref{eq:bell-weights}, the exact $(2+u)$-moment of the fibre-weight
vector is
\begin{align}
  \sum_\gamma q_{r,K,\gamma}^{\,2+u}
  &=
  \left(\frac1K\right)^{2+u}
  +
  \frac{r-1}{2r}K^2
  \left(\frac2{K^2}\right)^{2+u}
  +
  \left(\frac{K^2}{r}-K\right)
  \left(\frac1{K^2}\right)^{2+u}
  \notag\\
  &=
  K^{-2-2u}
  \left(
    K^u+\beta_r2^u+\frac1r-\frac1K
  \right).
  \label{eq:bell-high-moment}
\end{align}
Hence
\begin{equation}\label{eq:bell-high-entropy}
  H_{2+u}(q_{r,K})
  =
  2\log K
  -
  \frac{
    \log\!\left(
      K^u+\beta_r2^u+\frac1r-\frac1K
    \right)
  }{1+u}.
\end{equation}
Combining
\[
  H_{2+u}(Z_{\Phi_\pi})
  \le
  H_{2+u}(q_{r,K})
\]
from \eqref{eq:bell-entropy} with the one-copy estimate
\eqref{eq:one-copy-high},
\[
  H_{\min,2+u}(\Phi_\pi)
  \ge
  \log K
  -
  \frac{\log a+u\log b}{1+u},
\]
gives
\begin{equation}\label{eq:finite-high-gap}
  \Delta_{2+u}^\Omega(\Phi_\pi)
  \ge
  \frac{
    \log\!\left(
      K^u+\beta_r2^u+\frac1r-\frac1K
    \right)
    -2\log a-2u\log b
  }{1+u},
  \qquad u\ge0.
\end{equation}

To make this estimate uniform over all $u\ge0$, apply the weighted
arithmetic--geometric mean inequality with weights $3/4$ and $1/4$:
\begin{align}
  K^u+\beta_r2^u
  &=
  \frac34\frac{K^u}{3/4}
  +
  \frac14\frac{\beta_r2^u}{1/4}
  \notag\\
  &\ge
  \left(\frac{K^u}{3/4}\right)^{3/4}
  \left(\frac{\beta_r2^u}{1/4}\right)^{1/4}
  \notag\\
  &=
  e^{\alpha_r}
  \left(K^{3/4}2^{1/4}\right)^u.
  \label{eq:weighted-amgm}
\end{align}
Since $K=2rd$, one has
\[
  \frac1r-\frac1K\ge0.
\]
Thus \eqref{eq:finite-high-gap} implies
\begin{align}
  \Delta_{2+u}^\Omega(\Phi_\pi)
  &\ge
  \frac{
    \alpha_r-2\log a
    +
    u\left(
      \frac34\log K+\frac14\log2-2\log b
    \right)
  }{1+u}
  \notag\\
  &\ge
  \min\left\{
    \alpha_r-2\log a,\
    \frac34\log K+\frac14\log2-2\log b
  \right\},
  \qquad u\ge0,
  \label{eq:finite-high-uniform}
\end{align}
because the preceding fraction is a convex combination of the two
displayed constants.

Finally, at $p=\infty$, the Bell estimate
\[
  \|Z_{\Phi_\pi}\|\ge\frac1K
\]
from \eqref{eq:bell-quotient-majorization} gives
\[
  H_\infty(Z_{\Phi_\pi})
  \le
  \log K.
\]
Together with \eqref{eq:one-copy-infinity},
\[
  H_{\min,\infty}(\Phi_\pi)
  \ge
  \log K-\log b,
\]
this yields
\begin{equation}\label{eq:finite-infinity-gap}
  \Delta_\infty^\Omega(\Phi_\pi)
  \ge
  \log\frac{K}{b^2}.
\end{equation}
Since $K\ge2$,
\[
  \log\frac{K}{b^2}
  \ge
  \frac34\log K+\frac14\log2-2\log b.
\]

Combining the low-order estimate \eqref{eq:finite-low-gap}, the
uniform higher-order estimate \eqref{eq:finite-high-uniform}, and the
endpoint estimate above gives
\[
  \inf_{p\in[0,\infty]}
  \Delta_p^\Omega(\Phi_\pi)
  \ge
  \min\left\{
    D_{r,K,C},\
    \alpha_r-2\log a,\
    \frac34\log K+\frac14\log2-2\log b
  \right\}
  =
  \delta_{r,K,C}.
\]
This proves \eqref{eq:finite-common-gap}.  If
$\delta_{r,K,C}>0$, then \eqref{eq:witness-to-additivity-gap} gives
\[
  2H_{\min,p}(\Phi_\pi)
  -
  H_{\min,p}(\Phi_\pi^{\otimes2})
  \ge
  \Delta_p^\Omega(\Phi_\pi)
  >0
\]
simultaneously for every $p\in[0,\infty]$, completing the proof.
\end{proof}

\section{The free direct-product estimate}\label{sec:free-estimate}

The product Haagerup inequality is the analytic input. The additional
point specific to this model is that the same relative-operation fibres
which compress the Bell output also control coefficient collection in the
one-copy quadratic polynomial. We first prove the upper radius bound and
then show that all sectors can be asymptotically saturated together.

For the left regular representation $\lambda_{r,d}$ of
$\sG_{r,d}$, define
\begin{equation}\label{eq:regular-radius}
  \Gamma_{r,d}:=
  \sup_{\substack{A=A^*\in M_K(\C),\ \Tr A=0\\\|A\|_{\HS}=1}}
  \left\|\sum_{s,t\in\cS}a_{s,t}\lambda_{r,d}(s^{-1}t)\right\|.
\end{equation}
Let
\begin{equation}\label{eq:radius-constant}
  C_r:=\sqrt{16r^2-7r}.
\end{equation}

\subsection{Word sectors and coefficient multiplicities}

\begin{proposition}[Regular radius bound]\label{prop:free-radius}
For every $r,d\ge2$,
\begin{equation}\label{eq:free-radius-bound}
  \Gamma_{r,d}\le C_r=\sqrt{16r^2-7r}.
\end{equation}
\end{proposition}

\begin{proof}
Let $A=A^*$ satisfy
\[
  \Tr A=0,
  \qquad
  \|A\|_{\HS}=1.
\]
For each $1\le j\le r$, let
\[
  x_j^2
  :=
  \sum_{\substack{s,t\in\cS_j\\ s\ne t}}
  |a_{s,t}|^2,
\]
and, for $1\le j<\ell\le r$, let
\[
  y_{j\ell}^2
  :=
  \sum_{(s,t)\in
  (\cS_j\times\cS_\ell)\cup(\cS_\ell\times\cS_j)}
  |a_{s,t}|^2.
\]
These regions are pairwise disjoint.  Since the diagonal entries of
$A$ are not included in the quantities above,
\begin{equation}\label{eq:radius-sector-mass}
  \sum_{j=1}^r x_j^2
  +
  \sum_{j<\ell}y_{j\ell}^2
  \le
  \|A\|_{\HS}^2
  =
  1.
\end{equation}
The diagonal branch pairs contribute only to the identity element:
\[
  \sum_{s\in\cS}a_{s,s}\lambda_{r,d}(e)
  =
  (\Tr A)I
  =
  0.
\]
It therefore remains to estimate the nonidentity terms.  We separate
them according to whether the two branches lie in the same free
factor or in distinct factors.

\begin{itemize}
  \item \emph{Same-factor sectors.}
  For each $j$, set
  \[
    Q_j
    :=
    \sum_{\substack{s,t\in\cS_j\\ s\ne t}}
    a_{s,t}\lambda_{r,d}(s^{-1}t).
  \]
  By \Cref{lem:fibres}, the map
  \[
    (s,t)\longmapsto s^{-1}t
  \]
  is injective on
  $\{(s,t)\in\cS_j^2:s\ne t\}$.  Hence collecting equal group
  elements does not change the $\ell^2$ norm of the coefficients, which remains $x_j$. Every word occurring in $Q_j$ has multidegree $2e_j$.
  Therefore \Cref{lem:tensor-haagerup} gives
  \[
    \|Q_j\|
    \le
    (2+1)x_j
    =
    3x_j.
  \]

  \item \emph{Mixed-factor sectors.}
  For $j<\ell$, set
  \[
    Q_{j\ell}
    :=
    \sum_{(s,t)\in
    (\cS_j\times\cS_\ell)\cup(\cS_\ell\times\cS_j)}
    a_{s,t}\lambda_{r,d}(s^{-1}t).
  \]
  By \Cref{lem:fibres}, every relative operation in this sector has
  exactly two preimages,
  \[
    (s,t)
    \quad\text{and}\quad
    (t^{-1},s^{-1}).
  \]
  Thus, after collecting equal group elements, each resulting
  coefficient has the form
  \[
    c_g
    =
    a_{s,t}
    +
    a_{t^{-1},s^{-1}}.
  \]
  Using
  \[
    |z+w|^2
    \le
    2\bigl(|z|^2+|w|^2\bigr)
  \]
  on each mixed fibre and summing over the fibres yields
  \[
    \sum_g|c_g|^2
    \le
    2y_{j\ell}^2.
  \]
  Hence the $\ell^2$ norm of the collected coefficients is at most
  $\sqrt2\,y_{j\ell}$. Every word occurring in $Q_{j\ell}$ has multidegree
  $e_j+e_\ell$.  The corresponding tensorized Haagerup constant is
  \[
    (1+1)(1+1)=4,
  \]
  and therefore
  \[
    \|Q_{j\ell}\|
    \le
    4\sqrt2\,y_{j\ell}.
  \]
\end{itemize}

Since the identity contribution vanishes,
\[
  \sum_{s,t\in\cS}
  a_{s,t}\lambda_{r,d}(s^{-1}t)
  =
  \sum_{j=1}^r Q_j
  +
  \sum_{j<\ell}Q_{j\ell}.
\]
The triangle inequality and the preceding sector estimates give
\begin{align}
  \left\|
    \sum_{s,t\in\cS}
    a_{s,t}\lambda_{r,d}(s^{-1}t)
  \right\|
  &\le
  3\sum_{j=1}^r x_j
  +
  4\sqrt2\sum_{j<\ell}y_{j\ell}
  \notag\\
  &\le
  \left(
    9r+32\binom r2
  \right)^{1/2}
  \left(
    \sum_{j=1}^r x_j^2
    +
    \sum_{j<\ell}y_{j\ell}^2
  \right)^{1/2}
  \notag\\
  &\le
  \left(
    9r+32\binom r2
  \right)^{1/2}
  \notag\\
  &=
  \sqrt{16r^2-7r},
\end{align}
where the second inequality is Cauchy--Schwarz and the third uses
\eqref{eq:radius-sector-mass}.

Taking the supremum over all admissible $A$ gives
\[
  \Gamma_{r,d}
  \le
  \sqrt{16r^2-7r},
\]
which completes the proof.
\end{proof}

\subsection{Simultaneous saturation of the regular sectors}

The radius upper bound used Cauchy--Schwarz after summing the Haagerup sector estimates. Its sharpness therefore requires more than separate sharp examples for each sector: the sectors must approach their bounds on the same vectors. 

\begin{theorem}[Asymptotically sharp regular radius]\label{prop:radius-sharpness}
For $r,d\ge2$ and $k:=2d$,
\begin{equation}\label{eq:radius-sharpness-bounds}
  \begin{aligned}
    \Gamma_{r,d}^2&\le16r^2-7r,\\
    \Gamma_{r,d}^2&\ge
       r\frac{(3k-4)^2}{k(k-1)}
       +16r(r-1)\left(1-\frac1k\right)^2.
  \end{aligned}
\end{equation}
In particular, for every fixed $r\ge2$,
\begin{equation}\label{eq:radius-sharpness-limit}
  \lim_{d\to\infty}\Gamma_{r,d}^2=16r^2-7r.
\end{equation}
The same estimates give the uniform finite-parameter bound
\begin{equation}\label{eq:uniform-radius-error}
  0\le16r^2-7r-\Gamma_{r,d}^2\le\frac{16r^2}{d}
  \qquad(r,d\ge2).
\end{equation}
\end{theorem}

\begin{proof}
The upper bound follows from \Cref{prop:free-radius}. Fix $r,d\ge2$, put
$k:=2d$, and write
\begin{equation}\label{eq:regular-quadratic-map}
  \cQ_\infty(A)
  :=\sum_{s,t\in\cS}a_{s,t}\lambda_{r,d}(s^{-1}t),
  \qquad A=(a_{s,t})_{s,t\in\cS}.
\end{equation}
Let $L_j:=\sum_{s\in\cS_j}\lambda_{r,d}(s)$. Under
$\ell^2(\sG_{r,d})\cong\ell^2(\F_d)^{\otimes r}$, these self-adjoint
operators act on separate tensor factors.

For each $j$, let $E_j$ have entry $1/\sqrt{k(k-1)}$ on
$\{(s,t)\in\cS_j\times\cS_j:s\ne t\}$. For $j<\ell$, let $F_{j\ell}$ have
entry $1/(\sqrt2\,k)$ on
$(\cS_j\times\cS_\ell)\cup(\cS_\ell\times\cS_j)$.
All other entries are zero. Their supports are pairwise disjoint
and contain respectively $k(k-1)$ and $2k^2$ entries, so these
matrices are real symmetric, traceless, and Hilbert--Schmidt
orthonormal. Inverse closure and cross-factor commutation give
\begin{equation}\label{eq:radial-sector-polynomials}
  X_j:=\cQ_\infty(E_j)
      =\frac{L_j^2-kI}{\sqrt{k(k-1)}},
  \qquad
  Y_{j\ell}:=\cQ_\infty(F_{j\ell})
      =\frac{\sqrt2}{k}L_jL_\ell.
\end{equation}

Let $L$ be the adjacency operator of the $k$-regular Cayley tree of
$\F_d$, so that $L_j$ acts as $L$ on the $j$th tensor factor.
Set $t_k:=2\sqrt{k-1}$ and, for $n\ge2$, define
\[
  f_n(g):=
  \begin{cases}
    (k-1)^{-|g|/2},&1\le|g|\le n,\\
    0,&\text{otherwise}.
  \end{cases}
\]
The shell of radius $q\ge1$ contains $k(k-1)^{q-1}$ vertices, hence
\begin{equation}\label{eq:radial-norm}
  \|f_n\|^2=\frac{nk}{k-1}.
\end{equation}
Each non-root vertex has one parent and $k-1$ children. Consequently,
$(L-t_kI)f_n$ vanishes outside the shells of radii $0,1,n,n+1$;
their contributions to its squared norm are respectively
$k^2/(k-1),k,k,k$. Thus
\begin{equation}\label{eq:radial-residual-unnormalized}
  \|(L-t_kI)f_n\|^2=\frac{k^2}{k-1}+3k.
\end{equation}
For $\xi_n:=f_n/\|f_n\|$, this gives
\begin{equation}\label{eq:radial-residual}
  \|(L-t_kI)\xi_n\|^2
  =\frac{4k-3}{n}\longrightarrow0,
\end{equation}
and therefore
\begin{equation}\label{eq:radial-moments}
  \langle\xi_n,L\xi_n\rangle\longrightarrow t_k,
  \qquad
  \langle\xi_n,L^2\xi_n\rangle
  =\|L\xi_n\|^2\longrightarrow t_k^2.
\end{equation}

Now set $\zeta_n:=\xi_n^{\otimes r}$. By the product structure and
\eqref{eq:radial-sector-polynomials} and \eqref{eq:radial-moments},
\begin{align}
  \langle\zeta_n,X_j\zeta_n\rangle
  &\longrightarrow
    \frac{t_k^2-k}{\sqrt{k(k-1)}}
    =\frac{3k-4}{\sqrt{k(k-1)}}=:a_k,
    \label{eq:radial-same-limit}\\
  \langle\zeta_n,Y_{j\ell}\zeta_n\rangle
  &=\frac{\sqrt2}{k}\langle\xi_n,L\xi_n\rangle^2
    \longrightarrow
    4\sqrt2\left(1-\frac1k\right)=:b_k.
    \label{eq:radial-mixed-limit}
\end{align}
Thus all sector values are approached on the same unit vectors.
Set $\ell_0:=(r a_k^2+\binom r2 b_k^2)^{1/2}$ and define
\begin{equation}\label{eq:radius-test-matrix}
  A_0:=\frac1{\ell_0}
  \left(
    a_k\sum_{j=1}^rE_j
    +b_k\sum_{j<\ell}F_{j\ell}
  \right).
\end{equation}
Hilbert--Schmidt orthonormality gives $\|A_0\|_{\HS}=1$, and $A_0$
is Hermitian and traceless. Hence it is admissible in
\eqref{eq:regular-radius}. Moreover,
\[
  \langle\zeta_n,\cQ_\infty(A_0)\zeta_n\rangle
  \longrightarrow
  \frac{r a_k^2+\binom r2 b_k^2}{\ell_0}
  =\ell_0.
\]
Since $\|\zeta_n\|=1$, we obtain
$\Gamma_{r,d}\ge\|\cQ_\infty(A_0)\|\ge\ell_0$, and consequently
\begin{equation}\label{eq:radius-sharpness-lower}
  \Gamma_{r,d}^2
  \ge\ell_0^2
  =r\frac{(3k-4)^2}{k(k-1)}
   +16r(r-1)\left(1-\frac1k\right)^2.
\end{equation}
Together with \eqref{eq:free-radius-bound}, this proves
\eqref{eq:radius-sharpness-bounds}. Finally, since $k=2d\ge4$,
\[
  \begin{aligned}
    0
    &\le16r^2-7r-\Gamma_{r,d}^2\\
    &\le r\frac{15k-16}{k(k-1)}
      +16r(r-1)\left(\frac2k-\frac1{k^2}\right)\\
    &\le\frac{16r+32r(r-1)}{k}
      \le\frac{32r^2}{k}
      =\frac{16r^2}{d}.
  \end{aligned}
\]
This proves \eqref{eq:uniform-radius-error} and, for fixed $r$,
\eqref{eq:radius-sharpness-limit}.
\end{proof}

\section{Finite-dimensional transfer}
\label{sec:finite-dimensional-transfer}

The preceding section establishes the quadratic-radius estimate in the left regular representation of $\sG_{r,d}$.  To pass from this infinite-dimensional model to finite-dimensional channels, we need two consequences of strong convergence.  First, the pointwise convergence of group-algebra norms must hold uniformly over the compact coefficient class defining the quadratic radius.  Second, for any fixed finite collection of distinct group elements, their finite-dimensional images must eventually remain linearly independent. The latter ensures that the Bell quotient suffers no rank loss beyond the collisions already forced by the group relations.

\subsection{Uniform transfer and absence of additional rank loss}
\label{subsec:uniform-transfer}

Fix an integer $K\ge2$, let $G$ be a discrete group, and fix
$s_1,\ldots,s_K\in G$.  Let
\[
  \pi_n:G\to\U(m_n),
  \qquad n\in\N,
\]
be finite-dimensional unitary representations.  Define
\begin{equation}\label{eq:transfer-gamma-n}
  \Gamma_n
  :=
  \sup_{\substack{
    A=A^*\in M_K(\C),\ \Tr A=0\\
    \|A\|_{\HS}=1
  }}
  \left\|
    \sum_{i,j=1}^K
    a_{ij}\pi_n(s_i^{-1}s_j)
  \right\|,
\end{equation}
and
\begin{equation}\label{eq:transfer-gamma-infinity}
  \Gamma_\infty
  :=
  \sup_{\substack{
    A=A^*\in M_K(\C),\ \Tr A=0\\
    \|A\|_{\HS}=1
  }}
  \left\|
    \sum_{i,j=1}^K
    a_{ij}\lambda_G(s_i^{-1}s_j)
  \right\|.
\end{equation}

The following lemma records both finite-dimensional consequences of
strong convergence that will be used below.

\begin{lemma}[Uniform transfer and finite-set linear independence]
\label{lem:uniform-transfer}
Assume that $\pi_n$ strongly converges to the left regular
representation $\lambda_G$.

\begin{itemize}
  \item[(i)] The quadratic radii converge:
  \[
    \Gamma_n\longrightarrow\Gamma_\infty
    \qquad(n\to\infty).
  \]
  If instead $(\pi_n)_{n\ge1}$ is a sequence of random unitary
  representations strongly converging in probability to $\lambda_G$,
  then
  \[
    \Gamma_n\xrightarrow{\mathbb P}\Gamma_\infty.
  \]

  \item[(ii)] Let $T\subset G$ be any fixed finite set of distinct
  group elements.  Then, in the deterministic case, the matrices
  \[
    \bigl(\pi_n(g)\bigr)_{g\in T}
  \]
  are linearly independent for all sufficiently large $n$.  In the
  random case,
  \[
    \mathbb P\!\left[
      \bigl(\pi_n(g)\bigr)_{g\in T}
      \text{ is linearly independent}
    \right]
    \longrightarrow1.
  \]
\end{itemize}
\end{lemma}

\begin{proof}
We prove the two assertions separately.

\begin{itemize}
  \item[(i)] \emph{Uniform convergence of the quadratic radius.}
  Let
  \[
    \mathcal A_K
    :=
    \left\{
      A=A^*\in M_K(\C):
      \Tr A=0,\ \|A\|_{\HS}=1
    \right\}.
  \]
  Since $K$ is fixed, $\mathcal A_K$ is compact in the
  Hilbert--Schmidt norm.  For $A=(a_{ij})\in\mathcal A_K$, define
  \[
    F_n(A)
    :=
    \left\|
      \sum_{i,j=1}^K
      a_{ij}\pi_n(s_i^{-1}s_j)
    \right\|
  \]
  and
  \[
    F_\infty(A)
    :=
    \left\|
      \sum_{i,j=1}^K
      a_{ij}\lambda_G(s_i^{-1}s_j)
    \right\|.
  \]
  Since all $\pi_n(g)$ and $\lambda_G(g)$ are unitary, for
  $A,B\in\mathcal A_K$,
  \begin{align}
    |F_n(A)-F_n(B)|
    &\le
    \left\|
      \sum_{i,j=1}^K
      (a_{ij}-b_{ij})\pi_n(s_i^{-1}s_j)
    \right\|
    \notag\\
    &\le
    \sum_{i,j=1}^K|a_{ij}-b_{ij}|
    \notag\\
    &\le
    K\|A-B\|_{\HS}.
    \label{eq:transfer-lipschitz}
  \end{align}
  The same estimate holds for $F_\infty$.  Thus the entire family $(F_n)_n$, together with $F_\infty$, is uniformly $K$-Lipschitz on $\mathcal A_K$. Fix $\delta>0$, and choose a finite $\delta$-net
  \[
    \mathcal N_\delta\subset\mathcal A_K.
  \]
  For $B=(b_{ij})\in\mathcal N_\delta$, set
  \[
    z_B
    :=
    \sum_{i,j=1}^K
    b_{ij}s_i^{-1}s_j
    \in\C[G].
  \]
  For every $A\in\mathcal A_K$, choose
  $B\in\mathcal N_\delta$ with
  $\|A-B\|_{\HS}\le\delta$.  By
  \eqref{eq:transfer-lipschitz},
  \begin{align}
    |F_n(A)-F_\infty(A)|
    &\le
    2K\delta
    +
    \max_{B\in\mathcal N_\delta}
    \left|
      \|\pi_n(z_B)\|
      -
      \|\lambda_G(z_B)\|
    \right|.
    \label{eq:transfer-net-bound}
  \end{align}
  Taking the supremum over $A\in\mathcal A_K$ yields
  \begin{equation}\label{eq:transfer-gamma-bound}
    |\Gamma_n-\Gamma_\infty|
    \le
    2K\delta
    +
    \max_{B\in\mathcal N_\delta}
    \left|
      \|\pi_n(z_B)\|
      -
      \|\lambda_G(z_B)\|
    \right|.
  \end{equation}

  In the deterministic case, the maximum on the right tends to zero
  because $\mathcal N_\delta$ is finite and strong convergence holds
  for every fixed $z_B\in\C[G]$.  Hence
  \[
    \limsup_{n\to\infty}
    |\Gamma_n-\Gamma_\infty|
    \le
    2K\delta.
  \]
  Since $\delta>0$ is arbitrary,
  \[
    \Gamma_n\longrightarrow\Gamma_\infty.
  \]

  In the random case, fix $\varepsilon>0$ and choose $\delta>0$ so
  that $2K\delta<\varepsilon/2$.  By
  \eqref{eq:transfer-gamma-bound} and the union bound,
  \begin{align*}
    \mathbb P\!\left(
      |\Gamma_n-\Gamma_\infty|>\varepsilon
    \right)
    &\le
    \sum_{B\in\mathcal N_\delta}
    \mathbb P\!\left(
      \left|
        \|\pi_n(z_B)\|
        -
        \|\lambda_G(z_B)\|
      \right|
      >
      \frac{\varepsilon}{2}
    \right).
  \end{align*}
  Every summand tends to zero by strong convergence in probability,
  and the net is finite.  Therefore
  \[
    \Gamma_n\xrightarrow{\mathbb P}\Gamma_\infty.
  \]

  \item[(ii)] \emph{Absence of additional finite-set collapse.}
  Write
  \[
    T=\{g_1,\ldots,g_M\},
  \]
  where the $g_a$ are distinct, and let
  \[
    \mathbb S_T
    :=
    \left\{
      c=(c_g)_{g\in T}\in\C^T:
      \sum_{g\in T}|c_g|^2=1
    \right\}.
  \]
  For $c\in\mathbb S_T$, define
  \[
    G_n(c)
    :=
    \left\|
      \sum_{g\in T}c_g\pi_n(g)
    \right\|,
    \qquad
    G_\infty(c)
    :=
    \left\|
      \sum_{g\in T}c_g\lambda_G(g)
    \right\|.
  \]
  The regular representation separates the distinct group elements
  uniformly.  Indeed,
  \begin{align}
    G_\infty(c)
    &\ge
    \left\|
      \sum_{g\in T}
      c_g\lambda_G(g)\delta_e
    \right\|_2
    \notag\\
    &=
    \left\|
      \sum_{g\in T}c_g\delta_g
    \right\|_2
    =
    1.
    \label{eq:regular-linear-independence}
  \end{align}
  Moreover, for $c,d\in\mathbb S_T$,
  \begin{align*}
    |G_n(c)-G_n(d)|
    &\le
    \sum_{g\in T}|c_g-d_g|
    \le
    \sqrt{M}\,\|c-d\|_2,
  \end{align*}
  and the same estimate holds for $G_\infty$. Since $\mathbb S_T$ is compact, the same finite-net argument as in part~(i) gives
  \[
    \sup_{c\in\mathbb S_T}
    |G_n(c)-G_\infty(c)|
    \longrightarrow0
  \]
  in the deterministic case, and convergence to zero in probability in the random case.  Combining this with \eqref{eq:regular-linear-independence}, we obtain, in the deterministic case,
  \[
    \inf_{c\in\mathbb S_T}G_n(c)
    \ge
    \frac12
  \]
  for all sufficiently large $n$.  Hence no nonzero coefficient vector can satisfy
  \[
    \sum_{g\in T}c_g\pi_n(g)=0,
  \]
  so $(\pi_n(g))_{g\in T}$ is linearly independent. In the random case, the same argument gives
  \[
    \mathbb P\!\left[
      \inf_{c\in\mathbb S_T}G_n(c)
      \ge\frac12
    \right]
    \longrightarrow1,
  \]
  and therefore
  \[
    \mathbb P\!\left[
      (\pi_n(g))_{g\in T}
      \text{ is linearly independent}
    \right]
    \longrightarrow1.
  \]
\end{itemize}
\end{proof}

\subsection{Transfer to the Haar-orthogonal model}
\label{subsec:random-transfer}

We now apply the preceding lemma to the Haar-orthogonal tensor
representations introduced in \eqref{eq:random-representation}.  The
radius convergence gives the one-copy estimate required by
\Cref{thm:finite-criterion}, while finite-set linear independence shows
that the group-theoretic Bell quotient has no additional
representation-dependent rank collapse.

\begin{proposition}[Random transfer of the all-order gap and exact Bell rank]
\label{prop:random-transfer}
Fix $r,d\ge2$, let
\[
  \pi_N:\sG_{r,d}\to O(m_N)
\]
be the Haar-orthogonal tensor representation in
\eqref{eq:random-representation}, and let
\[
  \Phi_N:=\Phi_{\pi_N}.
\]
Let $\Gamma_{r,d}$ denote the quadratic radius of the left regular representation associated with the branch set $\cS$ in \eqref{eq:regular-radius}, and let
$C>\Gamma_{r,d}$. Then with the notation in \eqref{eq:finite-quadratic-radius}
\begin{equation}\label{eq:random-radius-event}
  \mathbb P\!\left[
    \Gamma_K\bigl((\pi_N(s))_{s\in\cS}\bigr)\le C
  \right]
  \longrightarrow1.
\end{equation}
Moreover,
\begin{equation}\label{eq:random-exact-rank}
  \mathbb P\!\left[
    \rank Z_{\Phi_N}=M_{r,K}
  \right]
  \longrightarrow1.
\end{equation}
Consequently,
\begin{equation}\label{eq:random-common-gap}
  \mathbb P\!\left[
    \inf_{p\in[0,\infty]}
    \Delta_p^\Omega(\Phi_N)
    \ge
    \delta_{r,K,C}
    \ \text{and}\
    \rank Z_{\Phi_N}=M_{r,K}
  \right]
  \longrightarrow1.
\end{equation}
\end{proposition}

\begin{proof}
By the Bordenave--Collins strong-convergence result recorded in \Cref{lem:bordenave-collins}, the random representations $\pi_N$ strongly converge in probability to the left regular representation $\lambda_{r,d}$. Apply \Cref{lem:uniform-transfer}(i) to the fixed branch set $\cS$.  This gives
\[
  \Gamma_K\bigl((\pi_N(s))_{s\in\cS}\bigr)
  \xrightarrow{\mathbb P}
  \Gamma_{r,d}.
\]
Since $C>\Gamma_{r,d}$,
\[
  \mathbb P\!\left[
    \Gamma_K\bigl((\pi_N(s))_{s\in\cS}\bigr)\le C
  \right]
  \longrightarrow1,
\]
which proves \eqref{eq:random-radius-event}.  On this event,
\Cref{thm:finite-criterion} yields
\begin{equation}\label{eq:random-gap-on-radius-event}
  \inf_{p\in[0,\infty]}
  \Delta_p^\Omega(\Phi_N)
  \ge
  \delta_{r,K,C}.
\end{equation}
We next identify the Bell rank.  Let
\[
  \mathcal T:=\cS^{-1}\cS.
\]
By \Cref{lem:fibres},
\[
  |\mathcal T|=M_{r,K}.
\]
Since $\mathcal T$ is a fixed finite set of distinct group elements,
\Cref{lem:uniform-transfer}(ii) gives
\[
  \mathbb P\!\left[
    (\pi_N(g))_{g\in\mathcal T}
    \text{ is linearly independent}
  \right]
  \longrightarrow1.
\]
On this event, the vectors
\[
  \eta_g
  =
  \frac1{\sqrt{m_N}}\vecop(\pi_N(g)),
  \qquad g\in\mathcal T,
\]
are linearly independent as well, because vectorization is a linear
isomorphism.  Hence their Gram matrix $R_{\pi_N}$ is positive
definite, and therefore
\[
  \rank R_{\pi_N}
  =
  |\mathcal T|
  =
  M_{r,K}.
\]
The exact Bell quotient in
\Cref{prop:bell-rank-defect} then gives
\[
  \rank Z_{\Phi_N}
  =
  \rank R_{\pi_N}
  =
  M_{r,K},
\]
which proves \eqref{eq:random-exact-rank}.

Finally, let
\[
  \mathcal E_N^{\mathrm{rad}}
  :=
  \left\{
    \Gamma_K\bigl((\pi_N(s))_{s\in\cS}\bigr)\le C
  \right\}
\]
and
\[
  \mathcal E_N^{\mathrm{rank}}
  :=
  \left\{
    \rank Z_{\Phi_N}=M_{r,K}
  \right\}.
\]
The preceding arguments show
\[
  \mathbb P(\mathcal E_N^{\mathrm{rad}})\to1,
  \qquad
  \mathbb P(\mathcal E_N^{\mathrm{rank}})\to1.
\]
Therefore
\[
  \mathbb P\!\left(
    \mathcal E_N^{\mathrm{rad}}
    \cap
    \mathcal E_N^{\mathrm{rank}}
  \right)
  \ge
  1-
  \mathbb P\!\left(
    (\mathcal E_N^{\mathrm{rad}})^c
  \right)
  -
  \mathbb P\!\left(
    (\mathcal E_N^{\mathrm{rank}})^c
  \right)
  \longrightarrow1.
\]
On this intersection,
\eqref{eq:random-gap-on-radius-event} and
\eqref{eq:random-exact-rank} hold simultaneously, proving
\eqref{eq:random-common-gap}.
\end{proof}

\section{Proof of the main results}\label{sec:finite-realizations}

\subsection{Proof of the simultaneous almost-one-bit theorem}

The proof of \Cref{thm:main} proceeds in two stages.  We first choose $r$ sufficiently large, and then, with this $r$ fixed, choose $d$ sufficiently large, so that $K=2rd$.  We then verify that these choices make all three terms in the certificate $\delta_{r,K,C}$ simultaneously at least $\log 2-\varepsilon$.

\begin{proof}[Proof of \Cref{thm:main}]
Let $\tau:=\min\{\varepsilon,1/2\}$ and choose
\begin{equation}\label{eq:one-bit-parameters}
  \begin{aligned}
    r&:=\max\{7,\lceil2/\tau\rceil\},
       &C&:=4r>C_r = \sqrt{16r^2 - 7r},\\
    T_\tau&:=\max\{4C^2/\tau,\ 2(1+C)^3\},\\
    d&:=\left\lceil\frac{T_\tau}{2r}\right\rceil,
       &K&:=2rd.
  \end{aligned}
\end{equation}
Then $r\ge7$, $d\ge2$, $K\ge T_\tau$,
$1/r\le\tau/2$, and $C^2/K\le\tau/4$.
Since $M_{r,K}/K^2\le(r+1)/(2r)$,
\begin{align*}
  D_{r,K,C}
  &\ge\log2-\log(1+1/r)-2\log(1+C^2/K)\\
  &\ge\log2-\tau.
\end{align*}
For the constants in \eqref{eq:beta-alpha}, $r\ge7$ gives
$\beta_r\ge12/7$ and
\[
  e^{4\alpha_r}=\frac{256\beta_r}{27}
       \ge\frac{1024}{63}>16.
\]
Consequently,
\[
  \alpha_r-2\log a\ge\log2-\tau/2.
\]
The second term defining $T_\tau$ implies
\[
  \frac34\log K+\frac14\log2-2\log b
  \ge\log2+\frac14\log b\ge\log2.
\]
Thus all three entries in \eqref{eq:all-order-certificate} are at
least $\log2-\tau\ge\log2-\varepsilon$.
With these parameters fixed, \Cref{prop:random-transfer} proves the
required common lower bound and the exact-rank event.

On that event,
\begin{equation}\label{eq:rank-ceiling}
  \begin{aligned}
    \inf_{p\in[0,\infty]}\Delta_p^\Omega(\Phi_N)
    &\le\Delta_0^\Omega(\Phi_N)\\
    &=2H_{\min,0}(\Phi_N)-\log M_{r,K}\\
    &\le\log\frac{K^2}{M_{r,K}}.
  \end{aligned}
\end{equation}
Moreover, $K\ge4r$ implies
\[
  \frac{M_{r,K}}{K^2}
  =\frac12+\frac1{2r}-\frac1K+\frac1{K^2}>\frac12,
\]
so the last logarithm is strictly less than $\log2$. This gives both sides of \eqref{eq:main-probability} on one event of probability tending to one. As $\varepsilon\downarrow0$, the choices satisfy
$r=O(\varepsilon^{-1})$, $C=O(\varepsilon^{-1})$, and
\[
  K<T_\tau+2r=O(\varepsilon^{-3}),
\]
which completes the proof.
\end{proof}

\subsection{The finite-output example}

\begin{proof}[Proof of \Cref{cor:numerical}]
Take
\[
  r=2,\qquad d=80,\qquad K=320,
\]
and set
\[
  C_{\mathrm{num}}
  :=
  \sqrt{50+\frac1{100}}
  >
  C_2
  =
  \sqrt{50}.
\]
For this choice,
\[
  M_{2,320}=76481,
\]
and, with the notation of \Cref{thm:finite-criterion},
\[
  a_{\mathrm{num}}
  :=
  1+\frac{C_{\mathrm{num}}^2}{K}
  =
  \frac{37001}{32000},
  \qquad
  b_{\mathrm{num}}
  :=
  1+C_{\mathrm{num}}.
\]
Hence
\[
  D_{2,320,C_{\mathrm{num}}}
  =
  \gamma_{320}.
\]
We first verify the strict positivity of $\gamma_{320}$ without relying on rounded logarithms.  Write
\[
  R
  :=
  \frac{320^4}
       {76481\,(370+1/100)^2},
  \qquad
  x:=R-1.
\]
Then
\[
  \gamma_{320}=\log R=\log(1+x).
\]
Exact integer comparison gives
\[
  1<R<\frac98,
  \qquad
  \frac{x}{1+x}>\frac7{5000}.
\]
Since
\[
  \log(1+x)\ge\frac{x}{1+x}
  \qquad(x\ge0),
\]
we obtain
\[
  \gamma_{320}
  >
  \frac7{5000}
  =
  1.4\times10^{-3}.
\]
It remains to verify that the other two entries of
\eqref{eq:all-order-certificate} are larger than
$D_{2,320,C_{\mathrm{num}}}$.  Here
\[
  \beta_2=1,
  \qquad
  a_{\mathrm{num}}<\frac76,
  \qquad
  b_{\mathrm{num}}<\frac{81}{10}.
\]
Since
\[
  e^{4\alpha_2}
  =
  \frac{256}{27}
  >
  \left(\frac74\right)^4,
\]
we have
\[
  \frac{e^{\alpha_2}}{a_{\mathrm{num}}^2}
  >
  \frac{7/4}{(7/6)^2}
  =
  \frac97
  >
  \frac98.
\]
Equivalently,
\[
  \alpha_2-2\log a_{\mathrm{num}}
  >
  \log\frac98.
\]
For the third entry, the inequality
\[
  2\cdot320^3>80^4
\]
implies
\[
  320^{3/4}2^{1/4}>80.
\]
Therefore
\[
  \frac{320^{3/4}2^{1/4}}
       {b_{\mathrm{num}}^2}
  >
  \frac{80}{(81/10)^2}
  =
  \frac{8000}{6561}
  >
  \frac98,
\]
and hence
\[
  \frac34\log320+\frac14\log2
  -2\log b_{\mathrm{num}}
  >
  \log\frac98.
\]

Since $R<9/8$,
\[
  \gamma_{320}
  =
  \log R
  <
  \log\frac98.
\]
Thus both higher-order entries of
$\delta_{2,320,C_{\mathrm{num}}}$ strictly exceed
$\gamma_{320}$, and therefore
\[
  \delta_{2,320,C_{\mathrm{num}}}
  =
  \gamma_{320}.
\]

Finally, since
\[
  C_{\mathrm{num}}
  >
  C_2
  \ge
  \Gamma_{2,80},
\]
\Cref{prop:random-transfer} applies with $C=C_{\mathrm{num}}$, completing the proof.
\end{proof}

\subsection{The non-random existence consequence}

\begin{proof}[Proof of \Cref{cor:nonrandom}]
Fix $r,d$ as in \eqref{eq:one-bit-parameters}.  The complete
$r$-partite graph $\Lambda_{r,d}$ is finite, and
\[
  A_{\Lambda_{r,d}}
  \cong
  \sG_{r,d}
\]
by \eqref{eq:raag-direct-product}.  Hence
\Cref{lem:magee-thomas} provides finite-dimensional unitary
representations
\[
  \rho_n:\sG_{r,d}\to\U(\ell_n)
\]
that strongly converge to the left regular representation.
By \Cref{lem:realification}, their realifications
\[
  \pi_n:=\rho_n^{\R}:
  \sG_{r,d}\to O(2\ell_n)
\]
retain this strong convergence.

We first transfer the quadratic-radius estimate.  By the deterministic
part of \Cref{lem:uniform-transfer},
\[
  \Gamma_K\bigl((\pi_n(s))_{s\in\cS}\bigr)
  \longrightarrow
  \Gamma_{r,d}.
\]
By \Cref{prop:free-radius},
\[
  \Gamma_{r,d}
  \le
  C_r
  <
  4r.
\]
Therefore, for all sufficiently large $n$,
\[
  \Gamma_K\bigl((\pi_n(s))_{s\in\cS}\bigr)
  \le
  C,
  \qquad
  C:=4r,
\]
where $C$ is exactly the parameter used in the proof of
\Cref{thm:main}.  The parameter estimates established there, together
with \Cref{thm:finite-criterion}, consequently give
\[
  \inf_{p\in[0,\infty]}
  \Delta_p^\Omega(\Phi_{\pi_n})
  \ge
  \log2-\varepsilon
\]
for all sufficiently large $n$. By \Cref{lem:uniform-transfer}(ii), applied to
\[
  \mathcal T:=\cS^{-1}\cS,
\]
the matrices $(\pi_n(g))_{g\in\mathcal T}$ are linearly independent for all sufficiently large $n$.  Hence $R_{\pi_n}$ is positive definite, and \Cref{prop:bell-rank-defect}, together with
$|\mathcal T|=M_{r,K}$ from \Cref{lem:fibres}, gives
\[
  \rank Z_{\Phi_{\pi_n}}
  =
  M_{r,K}.
\]
The $p=0$ argument in \eqref{eq:rank-ceiling} then yields the upper bound in \eqref{eq:nonrandom-gap}.

For the finite-output specialization $r=2$, $d=80$, and $K=320$, use $C_{\mathrm{num}} := \sqrt{50+\frac1{100}}$. The deterministic radius convergence and $C_{\mathrm{num}}>C_2\ge\Gamma_{2,80}$ show that the same finite criterion applies for all sufficiently large $n$.  Hence the calculation in the proof of \Cref{cor:numerical} gives
\[
  \inf_{p\in[0,\infty]}
  \Delta_p^\Omega(\Phi_{\pi_n})
  \ge
  \gamma_{320}
  >
  1.4\times10^{-3}.
\]
This completes the proof.
\end{proof}

\section{Discussion}\label{sec:discussion}

The mechanism developed above separates two effects of the
relative-operation map.  At one copy, bounded multiplicities of
nonidentity fibres keep the quadratic branch radius under control.
At two copies, the same group relations identify a positive proportion
of Bell branches and create a macroscopic exact kernel.  We close by
discussing three consequences of this separation: the origin and scope
of the one-bit scale, the output-dimension cost of the particular
purity--rank route used in the proof, and the relation with earlier
random-subspace and projection-induced constructions.

\subsection{The one-bit ceiling and the rank endpoint}

The Bell quotient exists for every finite real representation, although
its rank can in principle be smaller than the number of relative-operation
fibres because of additional linear dependencies.  Strong convergence
rules out such extra dependencies on the fixed set $\cS^{-1}\cS$.
By \Cref{prop:bell-rank-defect,prop:random-transfer}, for fixed $r,d$,
\[
  \Pr\!\left[\rank Z_{\Phi_N}=M_{r,K}\right]\longrightarrow1.
\]
Thus, with probability tending to one, the Bell kernel is exactly the
kernel forced by the group relations.  This is an exact
finite-dimensional rank statement.

This exact rank identifies the natural ceiling of the present common
Bell-witness mechanism.  On the above event,
\eqref{eq:rank-ceiling} gives
\[
  \inf_{p\in[0,\infty]}
  \Delta_p^\Omega(\Phi_N)
  \le
  \Delta_0^\Omega(\Phi_N)
  \le
  \log\frac{K^2}{M_{r,K}}.
\]
Since
\begin{equation}\label{eq:half-rank-limit}
  \frac{M_{r,K}}{K^2}
  =
  \frac{r+1}{2r}-\frac1K+\frac1{K^2}
  \longrightarrow
  \frac{r+1}{2r}
  \qquad
  (d\to\infty,\ r\ \text{fixed}),
\end{equation}
and the latter quantity tends to $1/2$ as $r\to\infty$, the corresponding
rank-entropy ceiling tends to $\log2$.  Combined with \Cref{thm:main},
which approaches the same value from below, this shows that $\log2$ is
the asymptotically sharp common gap scale for the present branch
geometry and the specified Bell input. 

The rank endpoint has a qualitatively different stability property from
the positive R\'enyi orders.  For orders bounded away from zero, the
R\'enyi entropy depends continuously on the output spectrum, so a
strict common gap with positive margin survives sufficiently small
perturbations of the channel.  At $p=0$, however,
$H_0(\rho)=\log\rank\rho$ depends on the presence of an exact kernel. For example, for any $0<\eta<1$, consider the channel obtained by adding a depolarizing component,
\[
  \Phi_\eta=(1-\eta)\Phi+\eta\mathcal R_K,
  \qquad
  \mathcal R_K(X)=(\Tr X)\frac{I_K}{K}.
\]
For every input state $\rho$ and every bipartite input state $\rho_{12}$,
\[
  \Phi_\eta(\rho)\ge \frac{\eta}{K}I_K,
  \qquad
  \Phi_\eta^{\otimes2}(\rho_{12})
  \ge \frac{\eta^2}{K^2}I_{K^2}.
\]
Hence every one-copy output has full rank $K$ and every two-copy output has full rank $K^2$. Consequently,
\[
  H_{\min,0}(\Phi_\eta)=\log K,
  \qquad
  H_{\min,0}(\Phi_\eta^{\otimes2})=2\log K,
\]
so the order-zero additivity defect vanishes exactly for every $\eta>0$ in this range. Thus even an arbitrarily small depolarizing perturbation completely removes the exact Bell kernel and destroys the rank-based $p=0$ violation.  In this sense, the order-zero endpoint is not perturbatively stable.  Consequently, extending the common violation all the way to $p=0$ requires more than an approximate spectral anomaly: it relies essentially on the exact finite-dimensional kernel enforced by the group relations.

\subsection{Dimension cost of the present purity--rank route}
\label{sec:certificate-scaling}

The choice in \Cref{thm:main} gives
\[
  K=O(\varepsilon^{-3})
\]
for a common gap at least $\log2-\varepsilon$.  It is useful to separate
this quantitative exponent from the underlying one-bit ceiling.  The
cubic scale is forced not by the Bell rank alone, but by the particular
way in which the present proof combines one-copy and two-copy
information.

More explicitly, the low-order part of the argument follows the chain
\[
  \text{quadratic radius}
  \ \Longrightarrow\
  \text{one-copy purity}
  \ \Longrightarrow\
  H_2
  \ \Longrightarrow\
  H_p\quad(0\le p\le2)
\]
at one copy, together with
\[
  \text{relative-operation fibres}
  \ \Longrightarrow\
  \text{Bell rank}
  \ \Longrightarrow\
  H_p\le\log\rank
\]
at two copies.  If the transferred quadratic radius is bounded by
$C\ge\Gamma_{r,d}$, then the exact Gram-radius identity gives
\[
  \Tr\Phi(X)^2
  \le
  \frac1K\left(1+\frac{C^2}{K}\right),
\]
and therefore
\[
  H_{\min,p}(\Phi)
  \ge
  \log K-\log\left(1+\frac{C^2}{K}\right),
  \qquad 0\le p\le2.
\]
On the other hand, the Bell fibre count gives
\[
  H_p(Z_\Phi)\le\log M_{r,K},
  \qquad p\ge0.
\]
Combining these two relaxations yields exactly the certificate
\begin{equation}\label{eq:discussion-low-certificate}
  D_{r,K,C}
  =
  \log\frac{K^2}
  {M_{r,K}(1+C^2/K)^2},
\end{equation}
already introduced in \eqref{eq:low-certificate}.  The estimates for
$p>2$ used in the full all-order criterion are additional; the cubic
dimension cost is already forced by the interval $0\le p\le2$.

\begin{proposition}[Cubic scaling within the present route]
\label{prop:certificate-scaling}
Let $r,d\ge2$, $K=2rd$, and $0<\varepsilon<\log2$.  Suppose that
\[
  C\ge\Gamma_{r,d},
  \qquad
  D_{r,K,C}\ge\log2-\varepsilon.
\]
Then
\begin{equation}\label{eq:certificate-scaling-lower}
  r\ge\frac{1}{2(e^\varepsilon-1)},
  \qquad
  K\ge\frac{9r(r-1)}{e^{\varepsilon/2}-1}.
\end{equation}
Consequently,
\[
  K=\Omega(\varepsilon^{-3}).
\]
\end{proposition}

\begin{proof}
Since $k=2d\ge4$, \eqref{eq:radius-sharpness-bounds} implies
\[
  C^2\ge\Gamma_{r,d}^2\ge9r(r-1),
  \qquad
  K\ge4r.
\]
Using the explicit form of $M_{r,K}$,
\begin{align*}
  \log2-D_{r,K,C}
  &=
  \log\!\left(
     1+\frac1r-\frac2K+\frac2{K^2}
  \right)
  +2\log\!\left(1+\frac{C^2}{K}\right)
  \\
  &\ge
  \log\!\left(1+\frac1{2r}\right)
  +2\log\!\left(1+\frac{C^2}{K}\right).
\end{align*}
Both terms are nonnegative.  Requiring their sum to be at most
$\varepsilon$ gives \eqref{eq:certificate-scaling-lower}.  Hence
$r=\Omega(\varepsilon^{-1})$, while the purity term then forces
\[
  K=\Omega\!\left(\frac{r^2}{\varepsilon}\right)
   =\Omega(\varepsilon^{-3}),
\]
which completes the proof.
\end{proof}

Accordingly, improving the output-dimension dependence cannot be achieved merely by refining constants inside the same two relaxations. One would need either finer one-copy information---for example a direct description or optimization of the output spectrum beyond its purity---or a different two-copy branch geometry.  The value $K=\Theta(\varepsilon^{-3})$ is therefore a limitation of the present proof architecture, not a lower bound for arbitrary channels or even for every possible analysis of the present channel family.

\subsection{Relation to random-subspace and projection-induced constructions}

The almost-one-bit scale at the von Neumann point was already established
by Belinschi, Collins, and Nechita for Haar-random Stinespring channels.
Their analysis identifies the limiting one-copy output body and optimizes
the minimum output entropy over it, while a maximally entangled input to
the product of a channel and its conjugate provides the two-copy witness.
For fixed output dimension $k$ and asymptotic input ratio $t\in(0,1)$,
their Bell-witness defect satisfies
\begin{equation}\label{eq:bcn-comparison}
  \mathcal V_{\mathrm{BCN}}(k,t)
  =
  -t\log t-(1-t)\log(1-t)+o(1),
  \qquad k\to\infty,
\end{equation}
so that at $t=1/2$ the violation approaches $\log2$
\cite[Theorem~6.3]{BelinschiCollinsNechita2016}.
Thus the almost-one-bit phenomenon at $p=1$ is not new to the present
work.  What is new here is that the same asymptotic scale is realized
by one finite-dimensional real channel and the same Bell input
simultaneously over the entire R\'enyi range.

The distinction becomes particularly clear at the rank endpoint.
For fixed $k$ and $t$, the limiting Bell spectrum in the random-subspace
model is \cite[Theorem~6.2]{BelinschiCollinsNechita2016}
\begin{equation}\label{eq:bcn-bell-spectrum}
  \left(
    t+\frac{1-t}{k^2},
    \left(\frac{1-t}{k^2}\right)^{\!\times(k^2-1)}
  \right).
\end{equation}
All entries are strictly positive, so this Bell mechanism does not
produce the exact rank defect needed at $p=0$.  In the present model,
by contrast, the group relations force an exact finite-dimensional
Bell kernel of positive asymptotic density.  This is the structural
ingredient that allows the almost-one-bit comparison to persist all
the way to the rank entropy.

At the level of channel classes, the present construction fits naturally
inside the projection-induced framework of Leung, Lovitz, and Wu
\cite{LeungLovitzWu2026}.  We use the unnormalized Choi convention:
for a linear map
\[
  \Theta:M_n(\C)\longrightarrow M_k(\C),
\]
its Choi matrix is
\[
  J_\Theta
  :=
  \sum_{a,b=1}^n
  \ket a\!\bra b
  \otimes
  \Theta(\ket a\!\bra b)
  =
  n\,(\id_{M_n}\otimes\Theta)
  \bigl(\proj{\Omega_n}\bigr),
\]
where
\[
  \ket{\Omega_n}
  :=
  \frac1{\sqrt n}\sum_{a=1}^n\ket a\otimes\ket a.
\]
With this convention,
\[
  \Theta(X)
  =
  \Tr_{\C^n}
  \!\left[
    J_\Theta\,(X^T\otimes I_k)
  \right],
  \qquad
  \Theta\ \text{is trace preserving}
  \iff
  \Tr_{\C^k}J_\Theta=I_n.
\]

Recall that a projection
\[
  P\in\mathcal B(\C^n\otimes\C^k),
  \qquad
  P_A:=\Tr_{\C^k}P>0,
\]
induces a trace-preserving Choi matrix through the inverse-marginal
normalization
\[
  J_P
  =
  (P_A^{-1/2}\otimes I_k)\,
  P\,
  (P_A^{-1/2}\otimes I_k).
\]
For the complementary channel \eqref{eq:complementary-channel}, define
\[
  T_W:\C^m\longrightarrow\C^m\otimes\C^{\cS},
  \qquad
  T_Wx
  :=
  \frac1{\sqrt K}
  \sum_{s\in\cS}W_s^Tx\otimes\ket s.
\]
A direct calculation gives
\begin{equation}\label{eq:choi-projection}
  J_\Phi
  =
  \frac1K
  \sum_{s,t\in\cS}
  (W_t^*W_s)^T\otimes\ket s\!\bra t
  =
  T_WT_W^*.
\end{equation}
Since $T_W^*T_W=I_m$, the operator $J_\Phi$ is itself an orthogonal
projection, and
\[
  \Tr_{\C^K}J_\Phi=I_m.
\]
Hence the inverse-marginal normalization is trivial in our case:
the present channels are projection-induced channels whose Choi
projections are already normalized.

The random ensembles, however, are different.  In
\cite{LeungLovitzWu2026}, the starting projection is Haar distributed
on the corresponding Grassmannian.  Here randomness is placed instead
on the group generators themselves: the $rd$ matrices
$O_{j,i}^{(N)}$ in \eqref{eq:random-representation} are independent
Haar orthogonal matrices, and their product Haar measure is pushed
forward to the Choi projection $J_{\Phi_N}$.  The resulting projection
is therefore highly structured rather than Grassmannian Haar.  In
particular, for real group branches,
\begin{equation}\label{eq:choi-blocks}
  [J_{\Phi_\pi}]_{s,t}
  =
  \frac1K\pi(t^{-1}s)^T
  =
  \frac1K\pi(s^{-1}t).
\end{equation}
Thus the Choi blocks retain the exact relative-operation information
coming from the representation of $\sG_{r,d}$.  The same words
$s^{-1}t$ govern the one-copy quadratic expressions, while the
corresponding group relations generate the exact Bell-branch
collisions at two copies.

This structural difference is also reflected in the quantifiers. For each fixed R\'enyi order $p$ in the ranges they treat, Leung, Lovitz, and Wu construct finite-dimensional projection-induced counterexamples, using a transpose-complement rank-defect mechanism at low orders and a product--conjugate Bell-state witness at higher orders \cite[Theorem~1.1 and Sections~1.1--1.2]{LeungLovitzWu2026}. Their statement does not require the same finite channel and the same
input to work uniformly over all orders.  Our present result instead has the simultaneous form
\[
  \begin{gathered}
    \forall\,\varepsilon\in(0,\log2)\quad
    \exists\,m,K\in\N\quad
    \exists\,\Phi:M_m(\C)\longrightarrow M_K(\C)
    \\[1mm]
    \text{such that}\qquad
    \forall\,p\in[0,\infty],\quad
    \Delta_p^\Omega(\Phi)\ge\log2-\varepsilon.
  \end{gathered}
\]
In this sense, the present model may be viewed as a structured projection-induced construction in which exact group relations are retained at the Choi level and converted into a common, almost-one-bit
Bell-witness gap across the full R\'enyi scale.

\section{Statements and Declarations}
\label{sec:statements-declarations}

\subsection{Use of Artificial Intelligence}

Inspired by the recent work \cite{ZhenZhuChenWang2026}, the authors investigated whether Haagerup-type inequalities could provide a common mechanism for R\'enyi-entropy nonadditivity.  In this process, a large language model (LLM) directly produced the initial technical argument based on the direct product of two free groups and a tensorized Haagerup inequality, leading to simultaneous nonadditivity for $p\in[0,2]$.  After analyzing and digesting this argument, the authors generalized the construction to direct products of an arbitrary number of free groups and developed a more general Bell-quotient analysis.  These extensions led to the stronger result established in the present manuscript: a single finite-dimensional real channel and a single Bell input exhibiting an almost-one-bit nonadditivity gap simultaneously for all $p\in[0,\infty]$.  The LLM was also used to assist with the exposition and editorial refinement of the manuscript.  All mathematical arguments, interpretations, references, and final editorial decisions were independently checked and approved by the authors, who take full responsibility for the scientific content of the paper.

\subsection{Acknowledgements and Concurrent Developments}

We thank Prof.~Peixue Wu for valuable discussions and comments on this work. Through private communication with him, we learned that, shortly after the appearance of \cite{LeungLovitzWu2026}, an additional argument was communicated to its authors by another researcher that extends their analysis to the remaining interval $p\in[1/4,3/4]$. Since this development is not yet publicly available, we do not rely on it anywhere in the present paper. In view of this information, however, we do not claim priority for resolving the interval $[1/4,3/4]$. This work was supported by the National Key R\&D Program of China (Grant No.~2024YFE0102500), the National Natural Science Foundation of China (Grant. No.~92576114, 12447107), the Guangdong Provincial Quantum Science Strategic Initiative (Grant No.~GDZX2403008, GDZX2503001), and the Guangdong Provincial Key Lab of Integrated Communication, Sensing and Computation for Ubiquitous Internet of Things (Grant No.~2023B1212010007).

\begingroup
\small
\newcommand{\etalchar}[1]{$^{#1}$}

\endgroup

\end{document}